\documentclass{aastex63} 
\usepackage{lineno}
\usepackage{tikz}
\usetikzlibrary{arrows}
\usepackage{amsmath}
\usepackage{comment}
\usepackage{amsthm} 
\usepackage{subfigure}
\usepackage{hyperref}

\usepackage{color}
\definecolor{SS}{rgb}{0.0, 0.0, 0.0}
\definecolor{MS}{rgb}{0.0, 0.0, 0.0}

\defcitealias{sengupta2021effects}{S21}

\shorttitle{The M-function}

\graphicspath{{./}{figures/}}

\begin{document}

\title{Effect of Thermal Emission in Isotropic Scattering Atmospheres: An Invariant-Embedding Extension of Chandrasekhar’s $H(\mu)$ -Function}

\correspondingauthor{Soumya Sengupta}
\email{ssengupta@ics.csic.es}
\email{soumyasenguptawb@gmail.com}

\author[0000-0002-7006-9439]{Soumya Sengupta}
\affiliation{Institut de Ci\`encies de I'Espai (ICE-CSIC), Campus UAB, Carrer de Can Magrans s/n, 08193 Cerdanyola del Vallès, Barcelona, Catalonia, Spain}
\author[0000-0002-2858-9385]{Manika Singla}
\affiliation{Physical Research Laboratory, Ahmedabad, 380009, India}
\affiliation{Planetary Atmospheres Group, Institute for Basic Science (IBS), Daejeon, South Korea}

\author[0000-0001-8357-340X]{Fikret Anlı}
\affiliation{Kahramanmaraş Sütçü İmam University, Faculty of Sciences, Physics Department, Kahramanmaraş/Turkey}


\begin{abstract}
Chandrasekhar’s $H(\mu)$-function forms the foundation of radiative transfer theory for semi-infinite, isotropically scattering atmospheres under external illumination. However, the classical formulation does not account for thermal emission from internal heat sources, which is essential in many astrophysical environments, including hot Jupiters, brown dwarfs, and strongly irradiated exoplanets, where re-radiated stellar energy significantly alters the emergent intensity. To address this limitation, we extend Chandrasekhar’s diffuse reflection framework by incorporating intrinsic thermal emission within the invariant-embedding formalism developed by \cite{bellman1967chandrasekhar}. In this approach, thermal emission enters as an embedded invariant contribution to the source function, leading to a generalized angular redistribution function $M(\mu)$ \citep{sengupta2021effects}. Building on the formulation of \cite{Chandrashekhar} and its recent extension, we derive the governing non-linear integral equations for $M(\mu)$ and express them in terms of the direction cosine $\mu$, the thermal emission coefficient $U(T)=B(T)/F$, and the single-scattering albedo $\tilde{\omega}_0$. High-precision numerical values of $M(\mu,U,\tilde{\omega}_0)$ are computed for $\mu\in[0,1]$, $U<0.7$, and $\tilde{\omega}_0<1$ using a stable iterative scheme based on Gaussian quadrature. In the limit of vanishing thermal emission, the formulation reduces to Chandrasekhar’s classical $H(\mu)$-function, validating the approach. As an application, we consider the ultra-short-period exoplanet K2-137b and identify the wavelength range $0.85$--$2.5~\mu$m where the model is most applicable, corresponding to the capabilities of \textit{JWST}, \textit{HST}, and \textit{ARIEL}.

\end{abstract}



\section{Introduction}\label{sec: intro}

Radiative transfer theory stands as a cornerstone of modern astrophysics, providing the mathematical and physical framework necessary to understand radiation propagation through scattering and absorbing media as shown in the classical works by \cite{Chandrashekhar} and \cite{mihalas1978stellar}. Its relevance extends across the astrophysical spectrum—from stellar interiors and photospheres to planetary and exoplanetary atmospheres. In particular, the foundational work presented in \cite{Chandrashekhar} on radiative transfer in a semi-infinite atmosphere, and the introduction of the Chandrasekhar $H(\mu)$-function, has had enduring significance. The $H(\mu)$-function, where $\mu$ is the direction cosine of the angle between the scattering plane and the pencil of radiation originally formulated to describe diffuse reflection in scattering media, remains central to contemporary models of planetary reflection and energy balance. 

In the astrophysical context, the theory of radiative transfer is indispensable—from the study of the Sun to the study of exoplanets. Classical solutions, while powerful, must be revisited in environments where internal heating or thermal emission cannot be ignored. \cite{bellman1967chandrasekhar} extended the original Chandrasekhar planetary problem to include internal heating sources in terms of invariant embedding to pave the way for incorporating thermal emission into the diffuse reflection framework.

This extension is particularly relevant in the context of irradiated giant exoplanets such as Hot Jupiters, where anomalously inflated radii point toward significant internal heat sources \citep{vigano2025inflated,komacek2017structure,batygin2010inflating}. For habitable worlds and Hot Jupiters alike, transmission and reflection spectroscopy are key observational tools, and the diffuse reflection theory of \cite{Chandrashekhar} continues to serve as the theoretical backbone in these analyses \citep{singla2023effect,singla2023new}. Even the theory is directly useful in modeling the planetary atmosphere \cite{madhusudhan2012analytic} as well as to determine the amount of heat redistribution through convection in the tidally locked hot Jupiters \cite{sengupta2023atmospheric}.

Chandrasekhar’s $H$-function has also found broader applications beyond planetary atmospheres. For instance, \cite{jablonski2019chandrasekhar} demonstrated its utility in modeling photoelectron transport, while \cite{jablonski2015chandrasekhar} presented a numerically stable and precise method for its evaluation in the isotropic scattering limit. The theoretical structure surrounding the $H$-function, including its integral properties, iterative solutions \citep{bosma1983efficient}, and moment expansions \citep{anli2021alternative}, has been rigorously studied. Its values for both isotropic and asymmetric scattering scenarios are tabulated in \cite{chandrasekharBreen1947radiative,chandrasekhar1948radiative}, with detailed discussion of the function’s moments and derivatives with better accuracy is  in \cite{das2007numerical}.

Furthermore, Chandrasekhar’s semi-infinite atmosphere theory finds direct application in areas ranging from planetary atmosphere modeling to secondary electron emission in ion-solid interactions \citep{dubus1986theoretical}. Nevertheless, traditional formulations of the $H$-function fail to capture scenarios where thermal emission and scattering are comparably significant.

This shortfall was first addressed in \cite{sengupta2021effects} (hereafter \citetalias{sengupta2021effects})
, where it was shown that thermal emission alters the source function in the semi-infinite diffuse reflection problem by contributing an additive term $U(T)f(\mu)$, where $f$ is a function of the direction cosine $\mu$. For the case of isotropic scattering, they established a specific functional form $f(\mu) = M(\mu)$, introducing what is now referred to as the $M$-function, analogous in form to the classical $H$-function but extended to account for thermal emission. Hence, the $M(\mu)$-function is introduced as a generalized angular distribution function for thermally active semi-infinite atmospheres, in which embeded invariant internal thermal emission  contributes directly to the radiation field and cannot be treated solely as an additive external source, while reducing exactly to the classical $H(\mu)$-function in the limit of vanishing emission.

The diffusely reflected specific intensity from a thermally emitting, isotropically scattering, semi-infinite atmosphere is given by \citetalias{sengupta2021effects}:

\begin{equation}\label{eq: specific intensity}
    I(0,\mu;\mu_0) = \frac{F}{4}\frac{\mu_0}{\mu+\mu_0} \left[4U(T) + \tilde{\omega}_0M(\mu_0)\right]M(\mu)
\end{equation}

Here, $\pi F$ is the incident flux along the direction $(-\mu_0, \phi_0)$, while the diffuse reflection is observed along $(\mu, \phi)$. The dimensionless \textit{thermal emission coefficient} is defined as $U(T) = B(T)/F$, where $B(T)$ is the Planck function. The $M$-function introduced in this context is a function of $\mu$, $U$, and the single scattering albedo $\tilde{\omega}_0$, i.e. $M(\mu,U,\tilde{\omega}_0)$ distinguishing it from the classical $H(\mu, \tilde{\omega}_0)$-function as given in \cite{Chandrashekhar}.

The physical interpretation and formal structure of $M(\mu, U, \tilde{\omega}_0)$ were further discussed in \cite{sengupta2022atmospheric} in relation to generalized $V(\mu)$ and $W(\mu)$ functions used for finite atmosphere models. However, those studies remained primarily qualitative. To date, no comprehensive effort has been made to quantify the $M$-function in its full parametric dependence for practical use in astrophysical applications.

In this work, we resolve this gap by formulating a rigorous derivation and complete numerical evaluation of the $M(\mu,U,\tilde{\omega}_0)$ function as defined in \citetalias{sengupta2021effects}. We adopt a methodology parallel to that employed by \cite{Chandrashekhar} in the classical treatment of the $H$-function: first deriving the integral theorems governing the $M$-function, then treating it as a function of three physical parameters—$\mu$, $U$, and $\tilde{\omega}_0$. We compute its values across a physically meaningful range of these parameters. These values are directly applicable in calculating reflected specific intensities via Eq.~\eqref{eq: specific intensity}, thus equipping the community with a practical tool for modeling thermal emission and scattering in planetary and exoplanetary atmospheres. The results presented here are fully consistent with, and extend beyond, the well-established literature on classical radiative transfer.

{The key novelty of this work is the dimensionless parameter $U(T)=B(T)/F$, which  represents the ratio of intrinsic thermal emission to incident irradiation. Physically, it sets the relative strength of locally generated photons compared to externally supplied radiation. While both thermal emission and scattering ultimately depend on the atmospheric state, within the present invariant-embedding framework $U(T)$ and the single-scattering albedo $\tilde{\omega}_0$ are treated as independent control parameters: $U(T)$ determines photon production, whereas $\tilde{\omega}_0$ governs their redistribution.}

The plan of the paper goes as follows. In Section ~ \ref{sec: Theorems of M-function}, we have derived the theorems applicable for $M(\mu,U,\tilde{\omega}_0)$ function introduced in \citetalias{sengupta2021effects}.Section ~ \ref{sec: value estimation} is devoted in estimating the values of M as a function of direction cosine $\mu=\cos\theta$, thermal emission co-efficient U and single scattering albedo $\tilde{\omega}_0$. For simplicity, this section is further divided into three subsections and discussed only emission in \ref{subsec: Only emission}, only scattering in \ref{subsec: Only Scattering} and simultaneous emission as well as scattering in \ref{subsec: emission+scattering}. In section ~ \ref{sec: consistency}, we discussed the consistency limit of our results with the previous studies. Finally we concluded by discussing the results, limitation and future works in Section~\ref{sec: discussion} follows with a general conclusion in Section~\ref{sec: conclusion}

\section{Theorems of M-Function}\label{sec: Theorems of M-function}
The present formulation is intended for isotropically scattering, semi-infinite atmospheres with embedded thermal emission treated within an invariant-embedding framework, and is not designed to reproduce classical Milne solutions except in well-defined asymptotic limits. Here $M(\mu)$ is treated as a generalized angular distribution function for coupled scattering and embedded thermal emission, rather than as a replacement for Chandrasekhar’s classical $H(\mu)$-function. We derive exact integral relations satisfied by the $M(\mu)$-function, analogous to the classical integral theorems of $H(\mu)$-function \cite{Chandrashekhar}. These relations characterize the global moments of the radiation field in a semi-infinite atmosphere where isotropic scattering and embedded thermal emission act simultaneously.

At first, we show the constrains on the zeroth-order moment of $M(\mu)$, governing the total reflected--emitted flux and defining a convergence condition for the coupled scattering--emission problem. Then we extend this analysis to higher-order moments and demonstrates how thermal emission modifies the mathematical structure of the solution by introducing a cubic dependence in $\mu$, which is absent in the classical scattering-only case. Together, these relations establish the domain of validity of the $M(\mu)$-function and provide practical consistency checks for numerical solutions.

Here we write the $M(\mu,U,\tilde{\omega}_0)$ function as $M(\mu)$ for simplifying the mathematical expressions. However the meaning remains the same. The functional form of $M(\mu)$-function is defined as \citetalias{sengupta2021effects},
\begin{equation}\label{eq: M-function}
M(\mu)= 1+2U(T)M(\mu)\mu \log(1+\frac{1}{\mu})+ \mu M(\mu) \frac{\tilde{\omega_0}}{2} \int_0^1 \frac{M(\mu')}{\mu+\mu'}d\mu'
\end{equation}

{where the second term on RHS represents the contribution of embedded thermal emission and the third term on RHS accounts for isotropic scattering.}

{Note that, here $U(T)$ acts as a source-strength parameter which controls the thermal photon to incident photon ratio, in contrast to $\tilde{\omega}_0$, which controls scattering.}

Here we define the moments of M-function as follows,
\begin{equation}\label{eq: M-moments}
    A_n = \int_0^1 \mu^n M(\mu) d\mu
\end{equation}

and the thermal emission contribution to the values of M-function,

\begin{equation}\label{eq: R-formula}
    R= 1+ 2U(T)\int_0^1 M(\mu) \mu log(1+\frac{1}{\mu}) d\mu
\end{equation}
\newtheorem{theorem}{Theorem}
\newtheorem{corollary}{Corollary}[theorem]

\begin{theorem}\label{thrm: THEOREM 1}

The integration of M-function can be written as,

\begin{equation}\label{eq: theorem 1}
\boxed{
\frac{\tilde{\omega_0}}{2}\int_0^1 M(\mu) d\mu ={\frac{\tilde{\omega_0}}{2}A_0}= 1-[1-\tilde{\omega}_0R]^\frac{1}{2}
}
\end{equation}
\end{theorem}

\begin{proof}
{We now derive a constraint on the zeroth moment $A_0$. The key step is to reduce the double integral term into a product of moments using symmetry.}

Multiplying eqn.\eqref{eq: M-function} by $\frac{\tilde{\omega_0}}{2}$ in both side and taking integration over $d\mu$ in the limit 0 to 1 we {obtain three contributions:
a constant term, a thermal emission term proportional to U(T), and a double integral term involving $M(\mu)M(\mu')$ as follows,}

\begin{equation*}
\begin{split}
\frac{\tilde{\omega_0}}{2} \int_0^1 M(\mu)d\mu = \frac{\tilde{\omega_0}}{2} +2U(T) \frac{\tilde{\omega_0}}{2} \int_0^1 M(\mu)\mu \log(1+\frac{1}{\mu}) d\mu +
(\frac{\tilde{\omega_0}}{2})^2 \int_0^1 \int_0^1 \frac{\mu}{\mu+\mu'}M(\mu')M(\mu)d\mu'd\mu
\end{split}
\end{equation*}

{The third term on RHS can be simplified by interchanging $\mu$ and $\mu'$ and taking the average of the two expressions. This symmetrization yields:}
\begin{equation*}
\begin{split}
\frac{\tilde{\omega_0}}{2} \int_0^1 M(\mu)d\mu =& \frac{\tilde{\omega_0}}{2} +2U(T) \frac{\tilde{\omega_0}}{2} \int_0^1 M(\mu)\mu \log(1+\frac{1}{\mu}) d\mu +
(\frac{\tilde{\omega_0}}{2})^2 \frac{1}{2} \int_0^1 \int_0^1 M(\mu')M(\mu)d\mu'd\mu\\
\implies \frac{1}{2}[(\frac{\tilde{\omega_0}}{2}) \int_0^1 M(\mu)d\mu]^2 -& [\frac{\tilde{\omega_0}}{2} \int_0^1 M(\mu)d\mu] + \frac{\tilde{\omega_0}}{2}\{1 +2U(T) \int_0^1 M(\mu)\mu \log(1+\frac{1}{\mu}) d\mu\}=0
\end{split}
\end{equation*}

This is a quadratic equation which has a solution as follows,
\begin{equation}\label{eq: theorem 1 plus-minus}
    \frac{\tilde{\omega_0}}{2}\int_0^1 M(\mu) d\mu = 1\pm[1-\tilde{\omega_0}\{ 1+ 2U(T)\int_0^1 M(\mu) \mu log(1+\frac{1}{\mu}) d\mu \}]^\frac{1}{2}
\end{equation}

In eqn.\eqref{eq: theorem 1 plus-minus} the left hand side uniformly converges to zero, when the single scattering albedo $\tilde{\omega}_0$ uniformly goes to zero. To satisfy the fact from the right side expression as well we will consider the negative sign rather than the positive sign. Thus it can be written as,
\
\begin{equation}\label{eq: theorem 1 elaborated}
     \frac{\tilde{\omega_0}}{2}\int_0^1 M(\mu) d\mu ={\frac{\tilde{\omega_0}}{2}A_0}= 1 - [1-\tilde{\omega_0}\{ 1+ 2U(T)\int_0^1 M(\mu) \mu log(1+\frac{1}{\mu}) d\mu \}]^\frac{1}{2}
\end{equation}

In this equation if we use eqn.\eqref{eq: R-formula} to replace R{, we} get {a} simplified form {which is equivalent to the final expression of \textit{Theorem} \ref{thrm: THEOREM 1} eqn.\eqref{eq: theorem 1}}.
\begin{equation*}
    \frac{\tilde{\omega_0}}{2}\int_0^1 M(\mu) d\mu = 1 - [1-\tilde{\omega_0}R]^\frac{1}{2}
\end{equation*}

{This relation represents a global flux constraint, modified by the presence of thermal emission through the parameter R. An alternate proof of this is given in appendix: \ref{appendix}}.
\end{proof}
\begin{corollary}\label{cor: 1.1}
The necessary condition for which the $M(\mu)$-function will be real can be written as, 
\begin{equation}\label{eq: cor 1.1}
    \frac{1-\tilde{\omega_0}}{\tilde{\omega_0}} \geqslant 2U(T)\int_0^1 M(\mu) \mu log(1+\frac{1}{\mu}) d\mu
\end{equation}

\end{corollary}

\begin{proof}
The right hand side of equation\eqref{eq: theorem 1} will be real only if,
\begin{equation*}
\begin{split}
&1  \geqslant \tilde{\omega_0}\{ 1+ 2U(T)\int_0^1 M(\mu) \mu log(1+\frac{1}{\mu}) d\mu \}\\
\implies
& \frac{1-\tilde{\omega_0}}{\tilde{\omega_0}} \geqslant 2U(T)\int_0^1 M(\mu) \mu log(1+\frac{1}{\mu}) d\mu
\end{split}
\end{equation*}
{This condition ensures that the argument of the square root in Eq. \eqref{eq: theorem 1} remains positive, thereby guaranteeing real and physically meaningful solutions for $M(\mu)$. It defines the allowed parameter space for the coupled scattering–emission problem.}
\end{proof}
\begin{corollary}\label{cor: 1.2}
An alternative integral equation can be formed as,
\begin{equation}\label{eq: cor 1.2}
\begin{split}
&\frac{1}{M(\mu)}= \frac{\tilde{\omega_0}}{2}\int_0^1 \frac{\mu'}{\mu+\mu'}M(\mu')d\mu' + [1-\tilde{\omega}_0R]^\frac{1}{2} - 2U(T)\mu \log(1+\frac{1}{\mu})
\end{split}
\end{equation}
\end{corollary}

\begin{proof}

{To obtain an alternative representation of the M-function, we rearrange the integral term in Eq. \eqref{eq: M-function} by separating them into symmetric components.}
\begin{equation*}
\begin{split}
\frac{\tilde{\omega_0}}{2}M(\mu)\int_0^1 \frac{\mu'}{\mu+\mu'}M(\mu')d\mu'&= \frac{\tilde{\omega_0}}{2}M(\mu)\int_0^1[1- \frac{\mu}{\mu+\mu'}]M(\mu')d\mu'\\
=&
\frac{\tilde{\omega_0}}{2}M(\mu)\int_0^1M(\mu')d\mu' -  \frac{\tilde{\omega_0}}{2}M(\mu)\mu\int_0^1\frac{M(\mu')}{\mu+\mu'}d\mu'\\
=&
M(\mu)(1-[1-\tilde{\omega}_0R]^\frac{1}{2})-M(\mu)+1+2U(T)M(\mu)\mu \log(1+\frac{1}{\mu})\\
=&
1- M(\mu)[1-\tilde{\omega}_0R]^\frac{1}{2}+2U(T)M(\mu)\mu \log(1+\frac{1}{\mu})
\end{split}
\end{equation*}
{Dividing both sides by $M(\mu)$ and rearranging we get the final form of} the eqn.\eqref{eq: cor 1.2}. {This form is useful for numerical evaluation, as it expresses M(µ) explicitly in terms of integral quantities.}
\end{proof}
\begin{theorem}\label{thrm: THEOREM 2}
\begin{equation}\label{eq: theorem 2}
\boxed{
    \frac{\tilde{\omega}_0}{6}= \frac{1}{2}[\frac{\tilde{\omega}_0}{2}\int_0^1 M(\mu) \mu d\mu]^2 + [1-\tilde{\omega}_0R]^\frac{1}{2}(\frac{\tilde{\omega}_0}{2})[\int_0^1 M(\mu)\mu^2 d\mu]-
(\frac{\tilde{\omega}_0}{2})2U(T) \int_0^1\mu^3M(\mu) \log(1+\frac{1}{\mu}) d\mu
}
\end{equation}
\end{theorem}

\begin{proof}
{We now extend the analysis to higher-order moments. In particular, we derive a relation involving the first, second and third moments of $M(\mu)$, which captures additional structure introduced by thermal emission.}

Multiplying eqn.\eqref{eq: M-function} by $\frac{\tilde{\omega}_0}{2}\mu^2$  {and integrating over $\mu \in [0,1]$, we again obtain three terms. }

\begin{equation}\label{eq: theorem 2 proving}
\begin{split}
\frac{\tilde{\omega}_0}{2} \int_0^1 M(\mu)\mu^2d\mu=& \frac{\tilde{\omega}_0}{2} \int_0^1\mu^2d\mu + 
2U(T)\frac{\tilde{\omega}_0}{2} \int_0^1\mu^3M(\mu) \log(1+\frac{1}{\mu}) d\mu 
+ (\frac{\tilde{\omega}_0}{2})^2 \int_0^1\mu^3 d\mu M(\mu)\int_0^1 \frac{M(\mu')}{\mu+\mu'}d\mu'
\end{split}
\end{equation}

{The most involved contribution arises from the double integral, which we simplify using the identity:} 
$$
\frac{\mu^3+\mu'^3}{\mu+\mu'} = \mu^2 -\mu\mu'+ \mu'^2
$$

{as follows. At first we interchange $\mu$ and $\mu'$ and take average of the third term on RHS of eqn. \eqref{eq: theorem 2 proving} to obtain,}
\begin{equation*}
    \begin{split}
        (\frac{\tilde{\omega}_0}{2})^2\frac{1}{2} \int_0^1 \int_0^1  \frac{\mu^3+\mu'^3}{\mu+\mu'} M(\mu)M(\mu')d\mu d\mu'
        &=
        (\frac{\tilde{\omega}_0}{2})^2\frac{1}{2} \int_0^1 \int_0^1  (\mu^2 - \mu\mu' + \mu'^2) M(\mu)M(\mu')d\mu d\mu'\\
        &=
        (\frac{\tilde{\omega}_0}{2})^2[\int_0^1 M(\mu)\mu^2 d\mu][\int_0^1 M(\mu)d\mu] - \frac{1}{2}[\frac{\tilde{\omega}_0}{2}\int_0^1 M(\mu) \mu d\mu]^2
    \end{split}
\end{equation*}

Now using Theorem \ref{thrm: THEOREM 1} we will get, 
\begin{equation*}
    \begin{split}
    &\frac{\tilde{\omega}_0}{2}[\int_0^1 M(\mu)\mu^2 d\mu][1-[1-\tilde{\omega}_0R]^\frac{1}{2}] - \frac{1}{2}[\frac{\tilde{\omega}_0}{2}\int_0^1 M(\mu) \mu d\mu]^2\\
    &= \frac{\tilde{\omega}_0}{2}[\int_0^1 M(\mu)\mu^2 d\mu] - \frac{\tilde{\omega}_0}{2}[\int_0^1 M(\mu)\mu^2 d\mu][1-\tilde{\omega}_0R]^\frac{1}{2}-\frac{1}{2}[\frac{\tilde{\omega}_0}{2}\int_0^1 M(\mu) \mu d\mu]^2
    \end{split}
\end{equation*}

Putting everything together in equation \eqref{eq: theorem 2 proving} we will get,
\begin{equation*}
         \frac{\tilde{\omega}_0}{6} + 
2U(T)\frac{\tilde{\omega}_0}{2} \int_0^1\mu^3M(\mu) \log(1+\frac{1}{\mu}) d\mu - \frac{\tilde{\omega}_0}{2}[\int_0^1 M(\mu)\mu^2 d\mu][1-\tilde{\omega}_0R]^\frac{1}{2}-\frac{1}{2}[\frac{\tilde{\omega}_0}{2}\int_0^1 M(\mu) \mu d\mu]^2=0
\end{equation*}

This expression is the same as of eq.\eqref{eq: theorem 2}. Hence proved \textit{Theorem} \ref{thrm: THEOREM 2}. {This relation shows that thermal emission introduces higher order (cubic) contributions in $\mu$, which are absent in the classical scattering-only case (e.g. H-function \citep{Chandrashekhar}).}
\end{proof}

{The above relations generalize the classical integral theorems of Chandrasekhar’s H-function \citep{Chandrashekhar} to the case of simultaneous scattering and thermal emission. The presence of U(T) modifies both the zeroth and higher-order moments, introducing additional nonlinear coupling between angular redistribution and intrinsic emission.}
\section{Estimation of the values of M-function:}\label{sec: value estimation}
In this section we will derive the values of $M(\mu,U, \tilde{\omega}_0)$ in a range of parameter values $\mu$, $\tilde{\omega}_0$ and U. From now on we will write the thermal emission co-efficient just as U for simplicity with the knowledge of how this function is defined. Hence the explicit expression of M-function can be written as,
\begin{equation}\label{eq: M-estimation}
M(\mu,U,\tilde{\omega}_0) = 1 + 2 U M(\mu,U,\tilde{\omega}_0) \mu \log(1+\frac{1}{\mu}) + \frac{\tilde{\omega}_0}{2}\mu M(\mu,U,\tilde{\omega}_0)\int_0^1 \frac{M(\mu',U,\tilde{\omega}_0)}{\mu+\mu'} d\mu'
\end{equation}

To progress further we will make it simple by classifying $M(\mu,U,\tilde{\omega_0})$ in three different regions as follows,
\begin{itemize}
    \item At first we will consider only emission and no scattering  ($U\neq0;\tilde{\omega_0}=0$). In such case the functional form $M(\mu,U,0)$ can be expressed analytically  as shown in eqn.\eqref{eq: M-function for only emission} and discussed in \ref{subsec: Only emission}.
    
    \item The next case considered is only scattering and no emission ($U=0;\tilde{\omega_0}\neq0$).  So the functional form $M(\mu,0,\tilde{\omega_0})$ will have a non-linear functional form eqn.\eqref{eq: M for only scattering} and can be solved only numerically. We explicitly discuss this in \ref{subsec: Only Scattering}.

    \item The final and most general case will be the case of both scattering and emission non zero. The functional form of $M(\mu,U,\tilde{\omega}_0)$ will be as eqn.\eqref{eq: M-estimation} and solved using the numerical methods as discussed in \ref{subsec: emission+scattering}
\end{itemize}

Note that eqn. \eqref{eq: M-estimation} consists of natural logarithm with base e as per the calculations presented in \citetalias{sengupta2021effects}. Hence all the following calculations are done in $log_e$, simply written as log to remain consistent with the literature, unless otherwise noted.
\subsection{Only emission:}\label{subsec: Only emission}
This is the simplest one among the three and can be solved analytically. For only emission case we can consider the scattering albedo,
$$\tilde{\omega}_0 = 0$$

Hence the equation of $M(\mu,U,0)$ function will be,

\begin{equation}\label{eq: M-function for only emission}
\begin{split}
&M(\mu,U,0) = 1 + 2 U M(\mu,U,0) \mu \log(1+\frac{1}{\mu})\\
\implies & M(\mu,U,0) = \frac{1}{1-2 U \mu \log(1+\frac{1}{\mu})}
\end{split}
\end{equation}

At this point we note that, eqn.\eqref{eq: M-function for only emission} has to give a positive value of $M(\mu)$ otherwise the emitting radiation will be negative which is unphysical. This puts the following condition on the denominator,
\begin{equation}\label{eq: emission upper limit}
1> 2 U\mu \log(1+\frac{1}{\mu})
\end{equation}

For the range of $\mu$ [0:1], we will get the maximum value of $\mu\log(1+\frac{1}{\mu})$ at $\mu = 1$. Hence the above condition will reduce into,
\begin{equation}\label{eq: emission upper limit value}
\begin{split}
1> 2 U \log(2)\\
\implies 1>1.386 U\\
\implies U<0.721
\end{split}
\end{equation}

 It means that for only emission case there is an upper limit of the blackbody flux to the irradiated flux ratio to get the real values of $M(\mu,U,0)$-function which can be written as , ${B(T)< 0.721} F$. In figure ~\ref{fig: M for omg=0} we have shown M as a function of $\mu$ for only emission case with varying U values upto 0.6 respecting the boundary set by the inequality condition eqn. \eqref{eq: emission upper limit value}. We will see a similar but more general condition for simultaneous emission and scattering as discussed in \ref{subsec: emission+scattering}

{It is worth noting that, in the absence of scattering ($\tilde{\omega}_0=0$), the emergent specific intensity within the present formulation can be written as eq.~(21) of \citetalias{sengupta2021effects}}
\begin{equation*}
I(0,\mu,\mu_0)=\frac{\mu_0}{\mu+\mu_0}\,B(T)\,M(\mu),
\end{equation*}

{Unlike the classical Milne pure--thermal--emission problem, the angular dependence does not vanish in this limit. Even though scattering of the incident flux is absent (i.e. $\tilde{\omega}_0 = 0$), the emergent intensity does not reduce to the isotropic result $I(0,\mu)=B(T)$. Instead, it retains an explicit dependence on the direction cosine $\mu$.}

{This behavior arises because, in the present diffuse--reflection framework \citetalias{sengupta2021effects}, thermal emission is not treated as a purely local source. Rather, it enters as an additive invariant contribution to the source function, normalized with respect to the incident flux. Consequently, setting $\tilde{\omega}_0=0$ removes scattering but does not eliminate the invariance constraint or the underlying reflection geometry, leading to a fundamentally different boundary--value problem than the classical formulation. We note that the convergence $I(0,\mu)\to B(T)$ in the grazing--angle limit $\mu\to0$ reflects the dominance of the local surface source function at arbitrarily large optical path lengths and does not imply equivalence with the Classical Milne formulation for general $\mu$. For an extensive discussion, we refer to \citetalias{sengupta2021effects}.}

\subsection{Only scattering:}\label{subsec: Only Scattering}
 In only scattering case we can consider the thermal emission contribution U=0 and thus the $M(\mu,0,\tilde{\omega}_0)$ will reduce into,
\begin{equation}\label{eq: M for only scattering}
M(\mu,0,\tilde{\omega}_0) = 1 + \frac{\tilde{\omega}_0}{2}\mu M(\mu,0,\tilde{\omega}_0)\int_0^1 \frac{M(\mu',0,\tilde{\omega}_0)}{\mu+\mu'} d\mu'
\end{equation}

This is nothing but Chandrasekhar's $H(\mu)$-function in isotropic scattering case as shown in eqn.\eqref{eq: H-function}. In fact in the no emission limit $U\to 0$, $M(\mu,0,\tilde{\omega}_0)\approx H(\mu,\tilde{\omega}_0)$ as discussed in section~\ref{sec: consistency}
. Clearly eqn.\eqref{eq: H-function} is a non-linear equation and its values can not be derived analytically. However the values of this function is well derived and tabulated in \cite{chandrasekharBreen1947radiative}.
 In this article these values are revisited 
 in case of no emission, for single scattering albedo values $\tilde{\omega}_0=0.1$ to 0.95. Only for perfect scattering case  $\tilde{\omega}_0=1.0;U=0$ case we need 3163 iterations for convergence (Not shown here), see section~\ref{sec: discussion}. {\color{SS}{In figure.\ref{fig: MvsOmega for fixed U}, the plot representing U=0 shows the variation of $M(\mu,0,\tilde{\omega}_0)$ with respect to $\tilde{\omega}_0$ for fixed $\mu$ s. Its variation with $\mu$ for fixed $\tilde{\omega}_0$ are also plotted by dashed lines in each plots of figure~\ref{fig: Comparison plots for M function}}}

\subsection{Simultaneous emission and scattering}\label{subsec: emission+scattering}

This is the most general case to be discussed in this article. The derived values exactly match with the previously derived values at the limiting conditions of only emission and only scattering. Here we solved the general equation of M-function as given in eqn. \eqref{eq: M-estimation}. Clearly this is a three variable non-linear equation which can only be solved numerically. The integral equation has been solved using the Gauss-Legendre quadrature method for 100 points, with convergence tolerance $10^{-8}$. All the numerical simulations has been done using the open source python packages \texttt{numpy} {\cite{2020NumPy-Array}} and \texttt{Scipy} {\cite{2020SciPy-NMeth}}. For numerical simulations we limit ourselves with maximum iterations 1000 which gives us quite satisfactory results.

While varying the $\tilde{\omega}_0$ we remind that the values of M will be real only within the regime constrained by Corollary\ref{cor: 1.1}. Taking account of that we calculated the M-values for $\tilde{\omega}_0$ = 0.1, 0.2,0.3,0.4,0.5,0.6,0.7,0.8,0.9,0.95 
respectively along with the maximum possible U (named $U_{max}$) for the convergence. It can be seen from fig.\ref{fig: Comparison plots for M function} that, $U_{max}$ for real $M(\mu)$, decreases with increasing $\tilde{\omega}_0$ satisfying the condition of corollary \ref{cor: 1.1}. The values of $M(\mu,U,\tilde{\omega}_0)$ are plotted in figure \ref{fig: Comparison plots for M function}

\begin{figure}[htbp]
    \centering
    \subfigure[$\tilde{\omega}_0=0.0;U_{max}=0.6$\label{fig: M for omg=0}]{\includegraphics[width=0.3\textwidth]{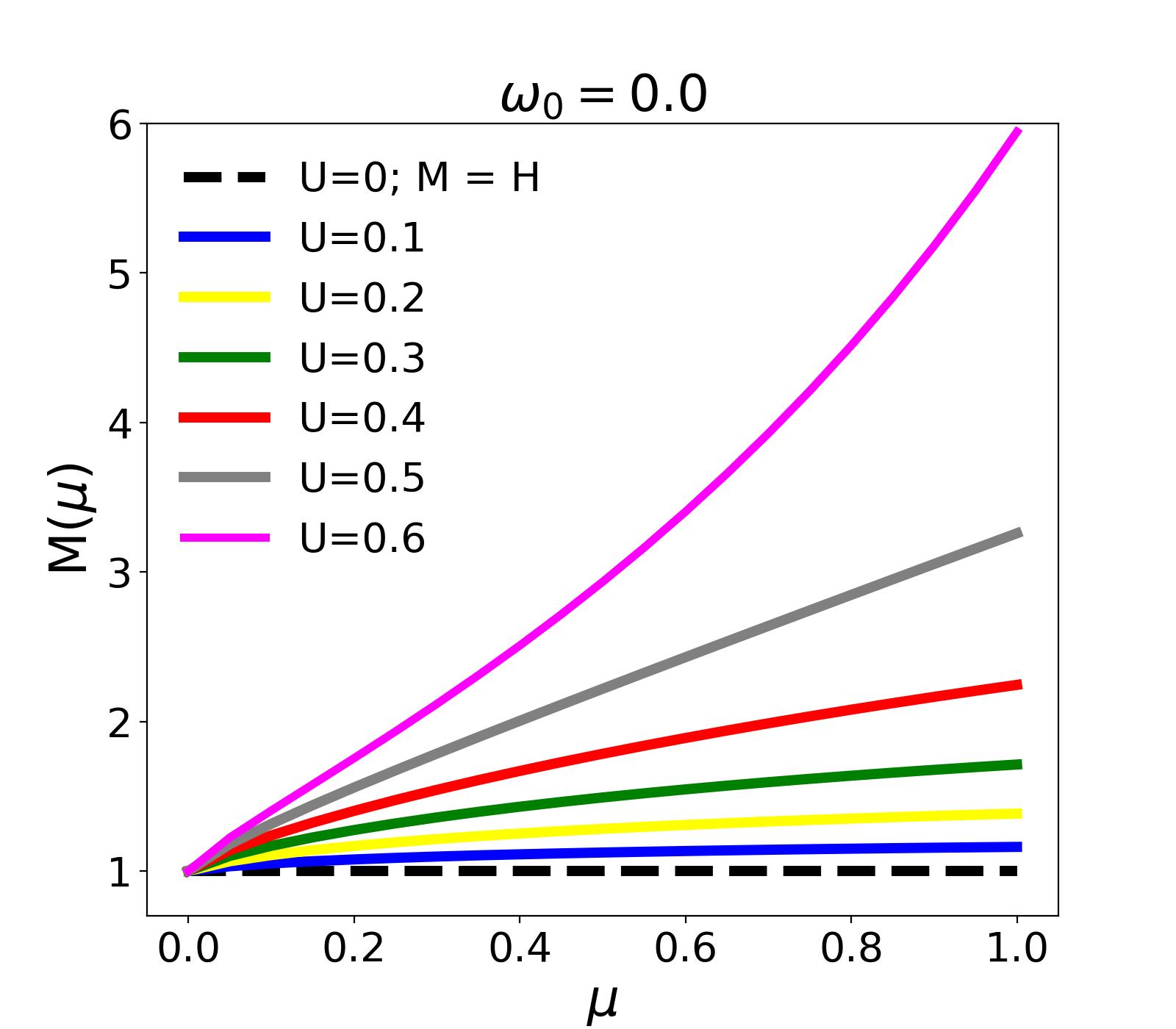}}
    \subfigure[$\tilde{\omega}_0=0.1;U_{max}=0.5$]{\includegraphics[width=0.3\textwidth]{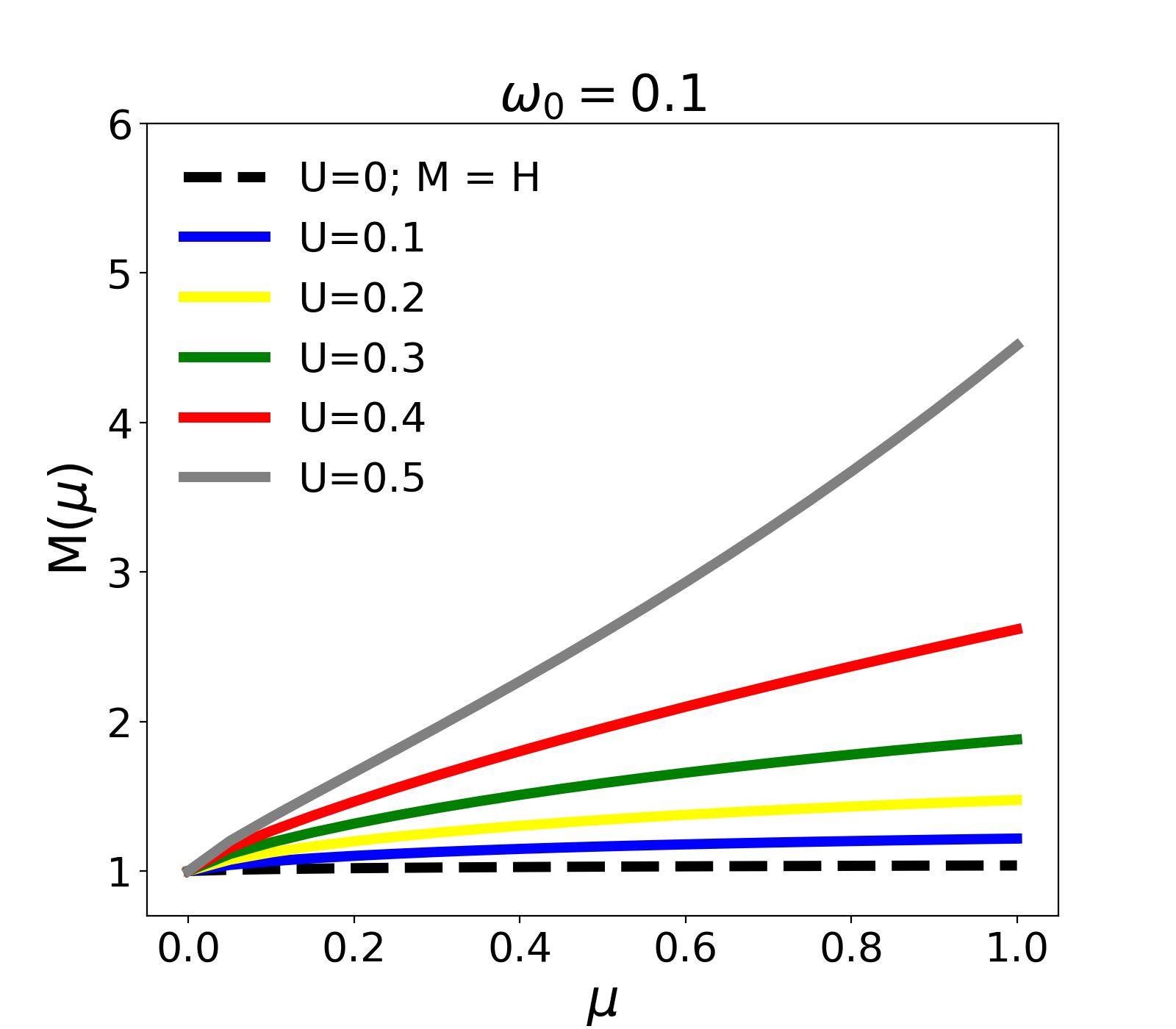}}
    \subfigure[$\tilde{\omega}_0=0.2;U_{max}=0.45$]{\includegraphics[width=0.3\textwidth]{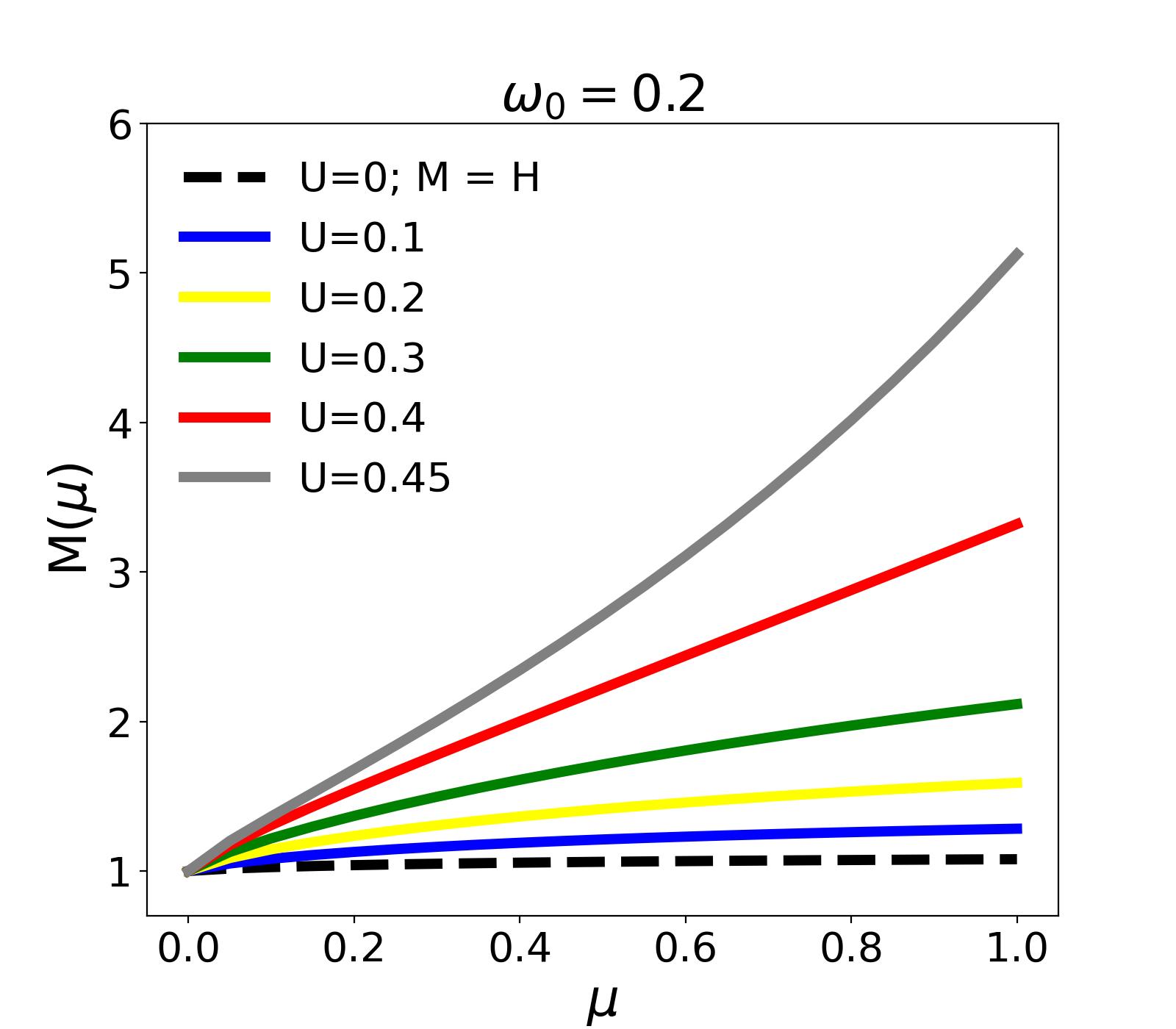}}
    
    \subfigure[$\tilde{\omega}_0=0.3;U_{max}=0.35$]{\includegraphics[width=0.3\textwidth]{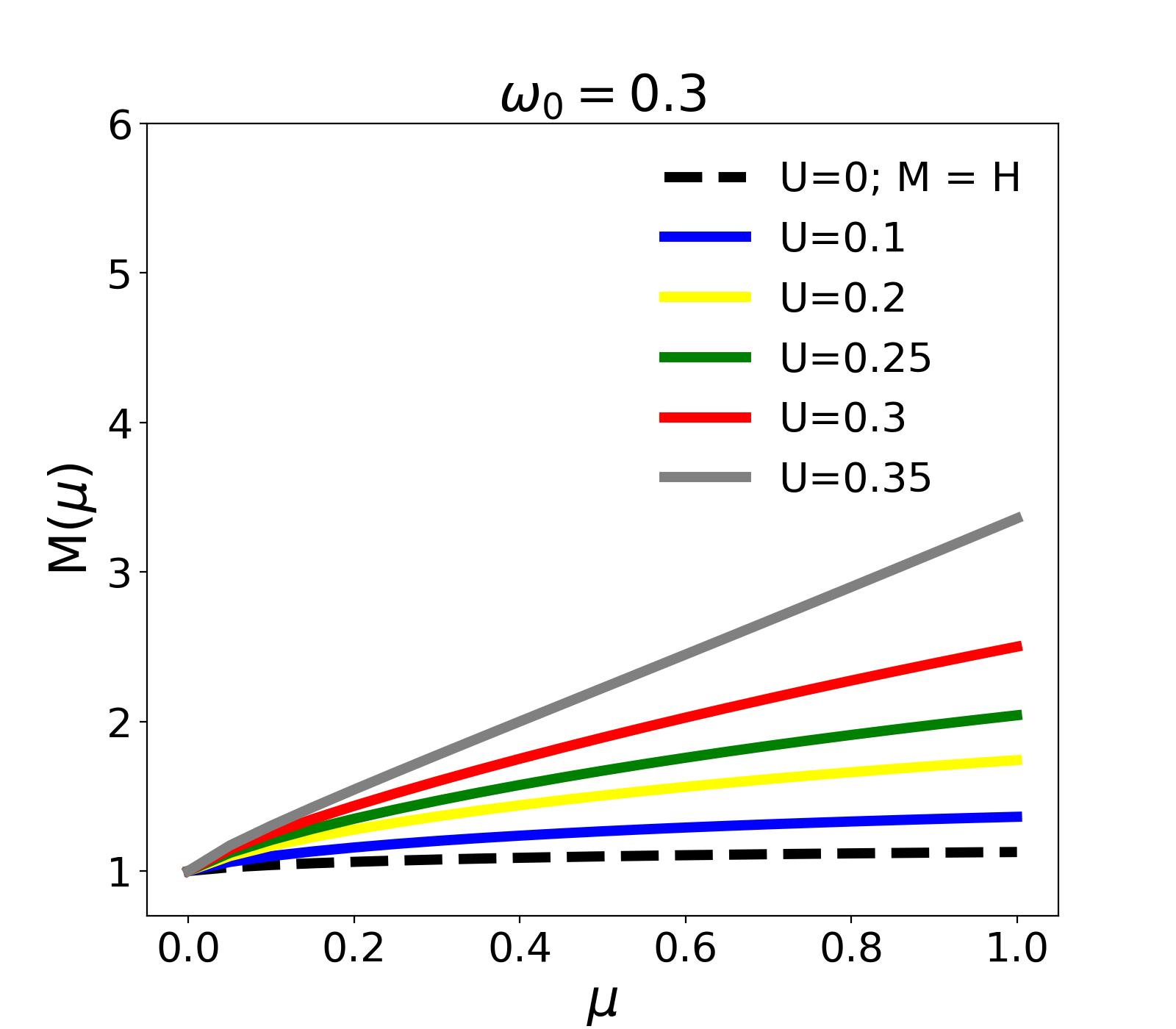}}
    \subfigure[$\tilde{\omega}_0=0.4;U_{max}=0.3$]{\includegraphics[width=0.3\textwidth]{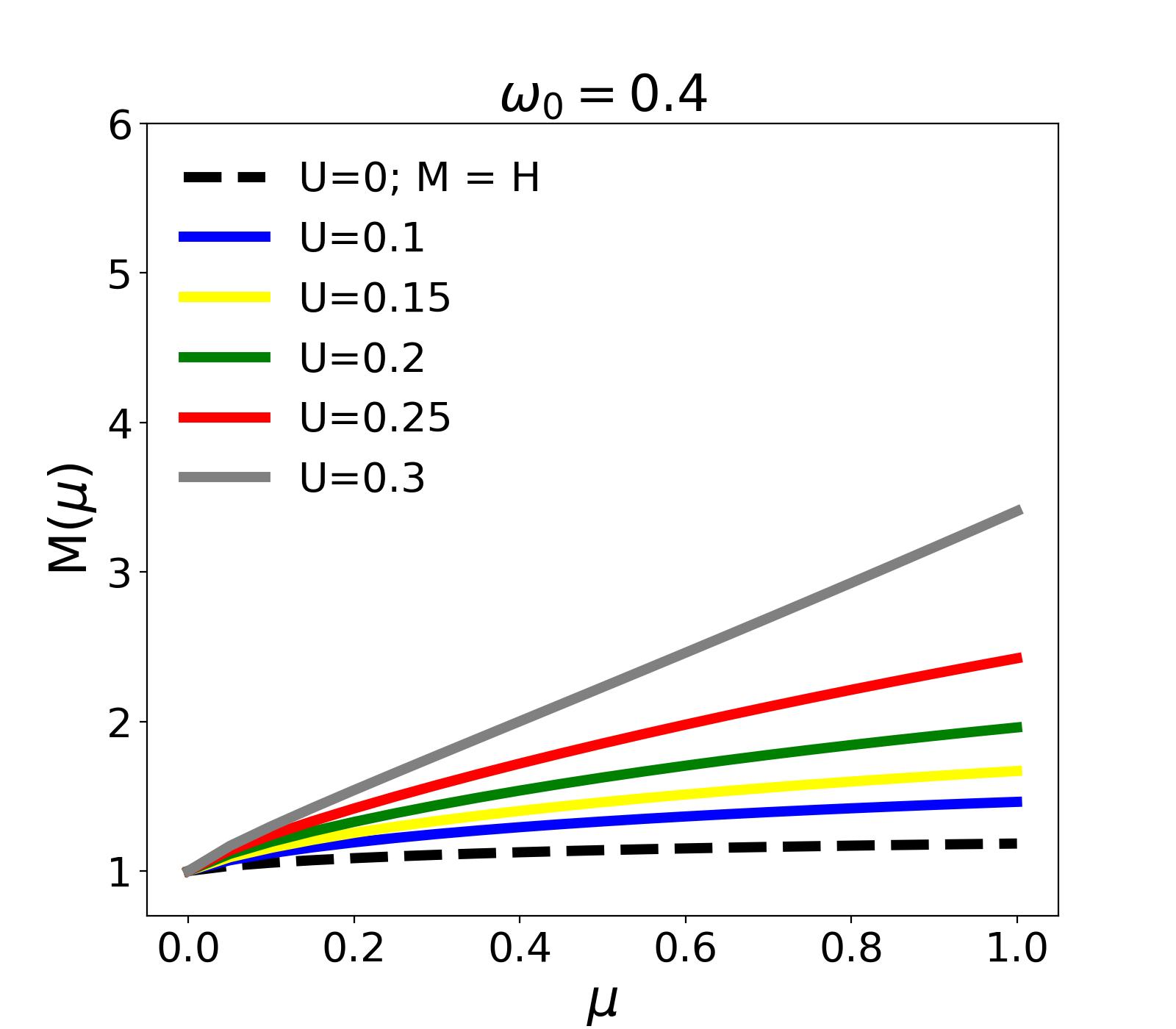}}
    \subfigure[$\tilde{\omega}_0=0.5;U_{max}=0.25$]{\includegraphics[width=0.3\textwidth]{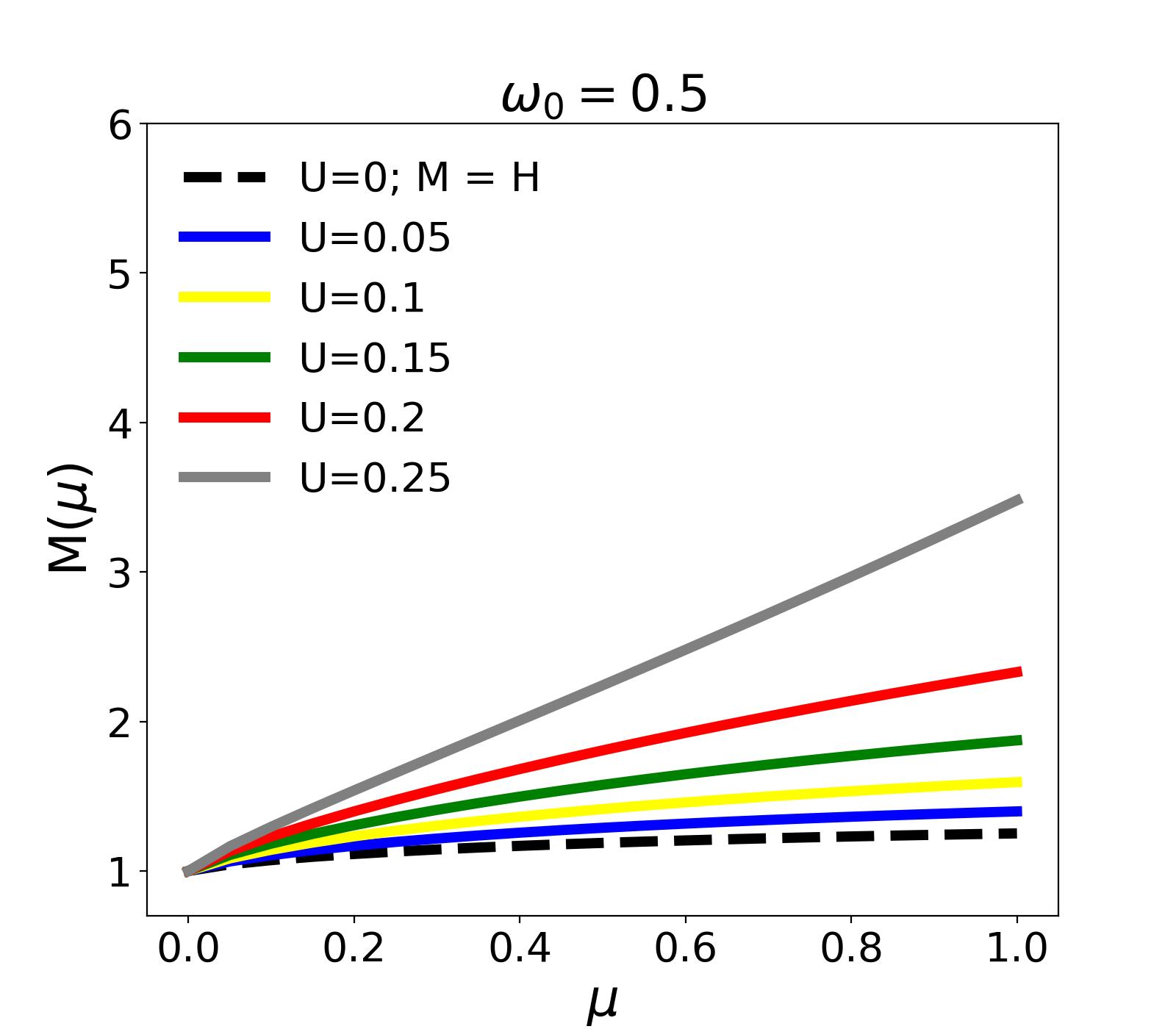}}
    
    \subfigure[$\tilde{\omega}_0=0.6;U_{max}=0.2$]{\includegraphics[width=0.3\textwidth]{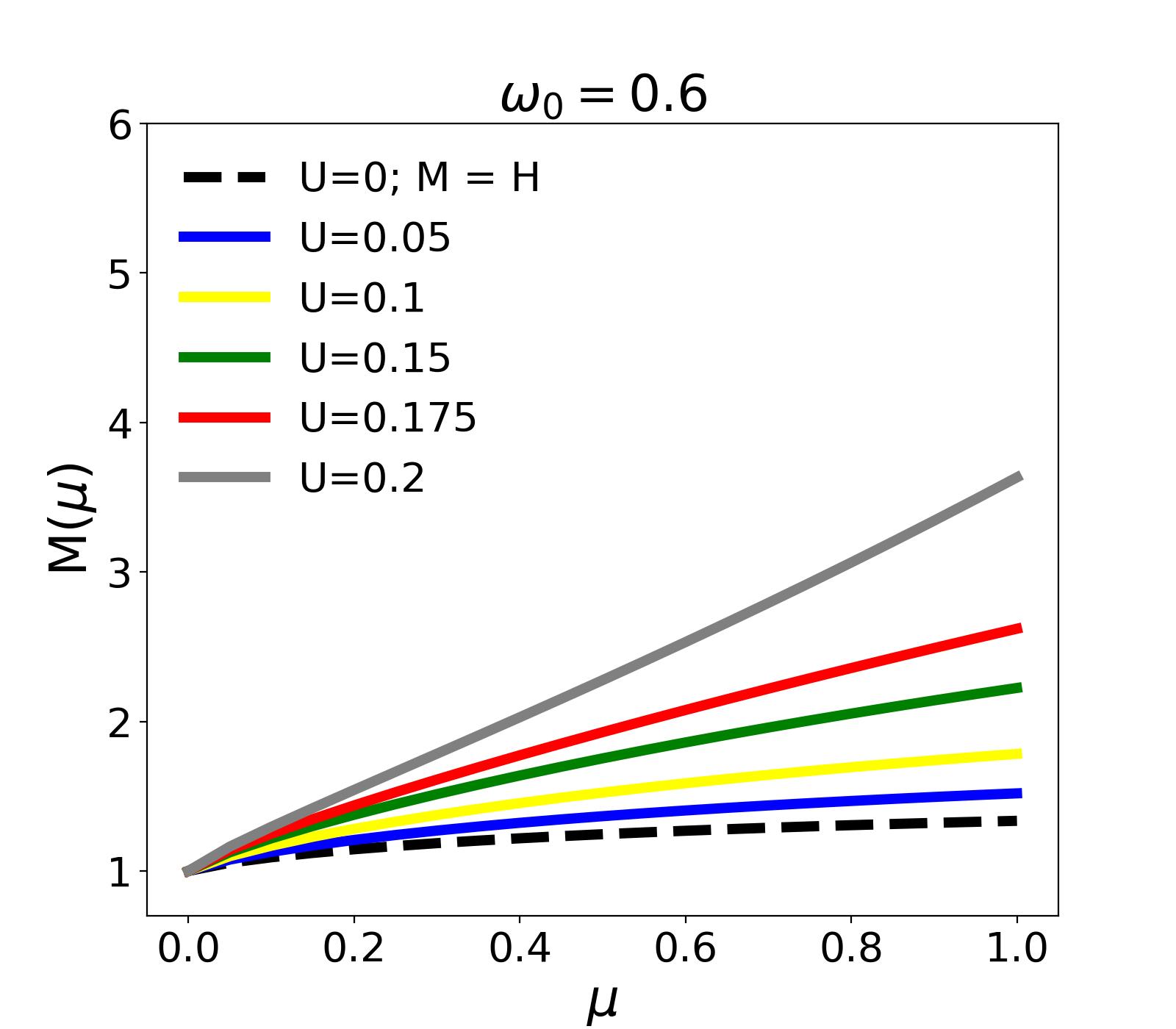}}
    \subfigure[$\tilde{\omega}_0=0.7;U_{max}=0.1$]{\includegraphics[width=0.3\textwidth]{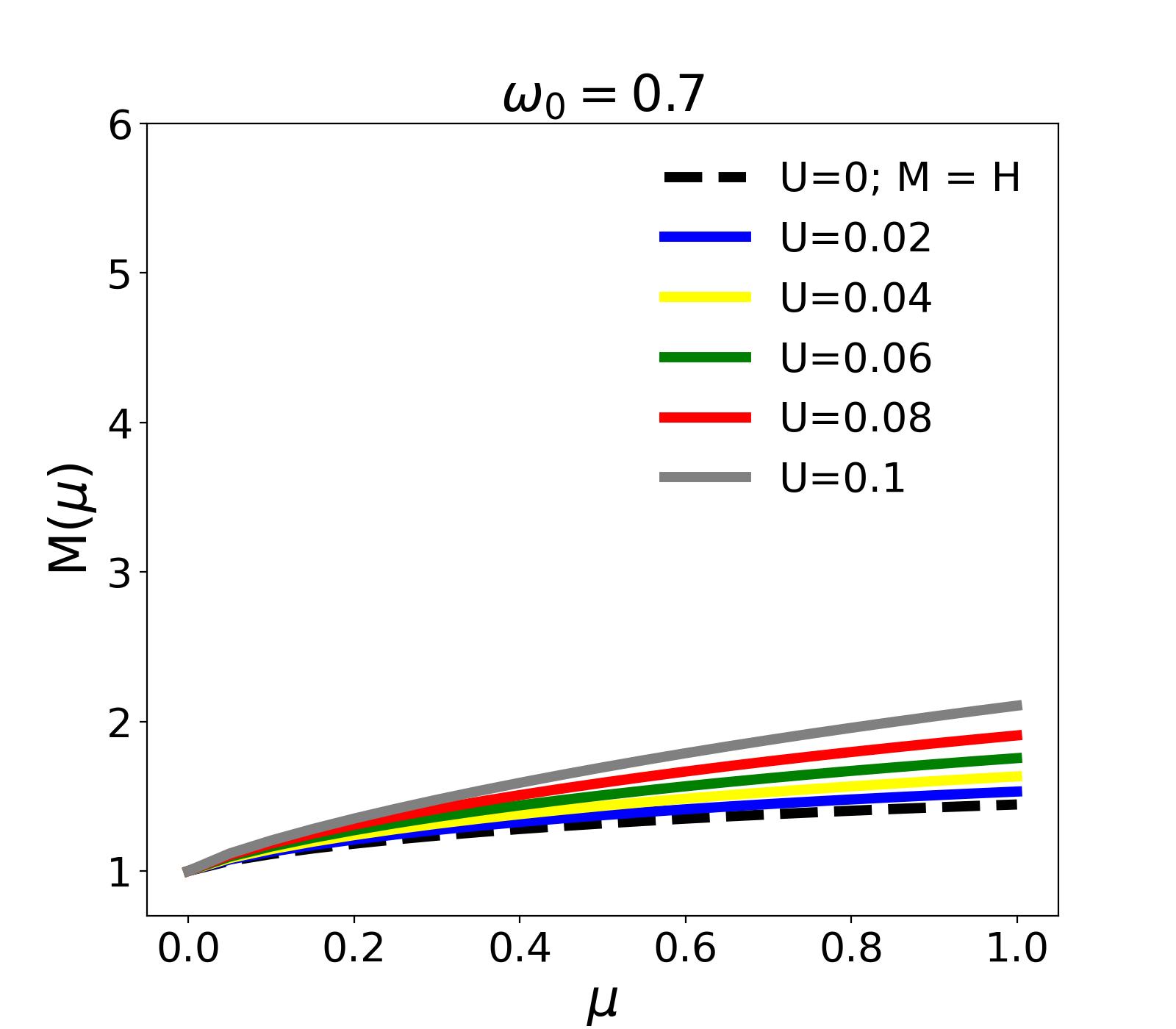}}
    \subfigure[$\tilde{\omega}_0=0.8;U_{max}=0.09$]{\includegraphics[width=0.3\textwidth]{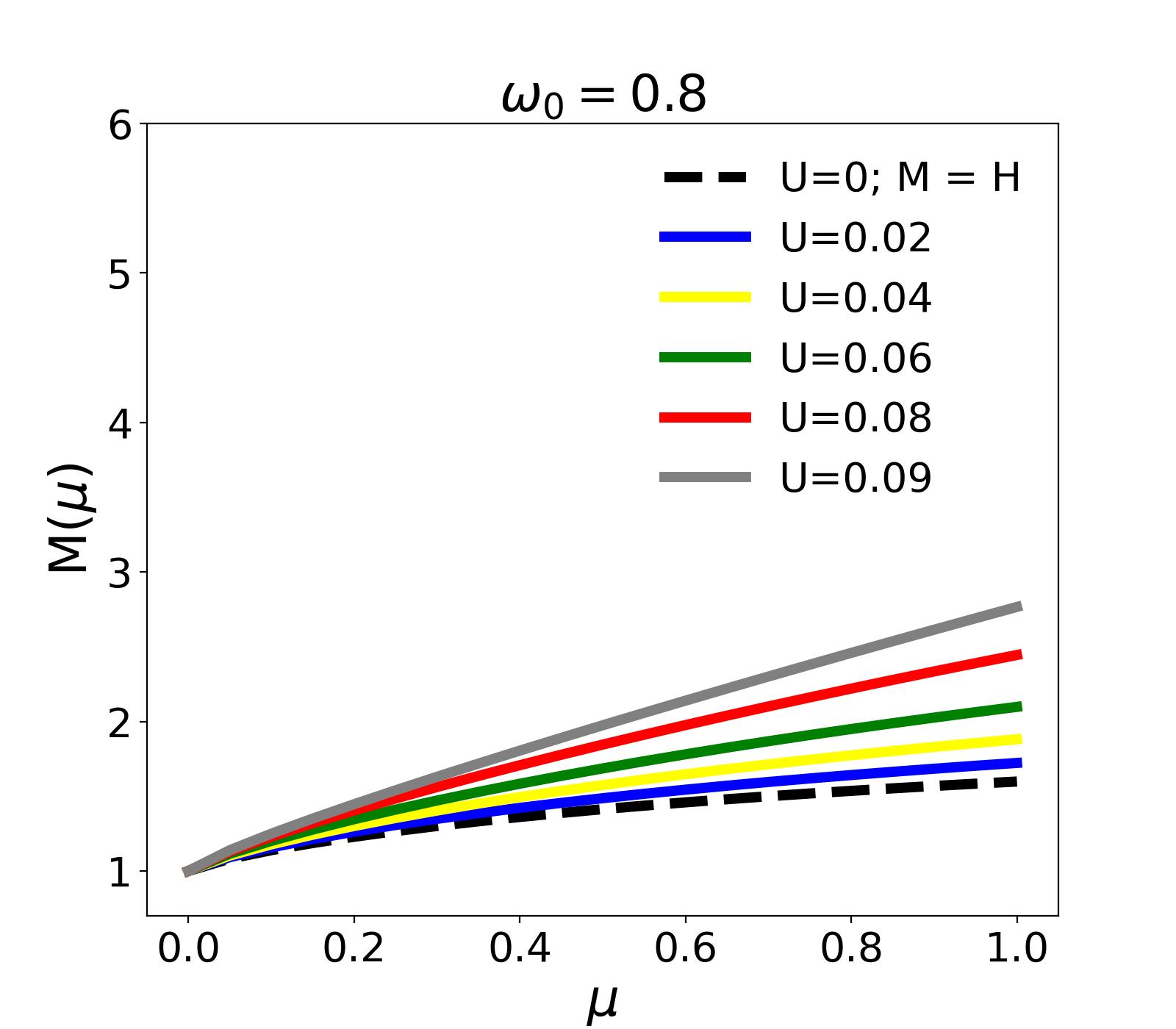}}
    
    \subfigure[$\tilde{\omega}_0=0.9;U_{max}=0.04$]{\includegraphics[width=0.3\textwidth]{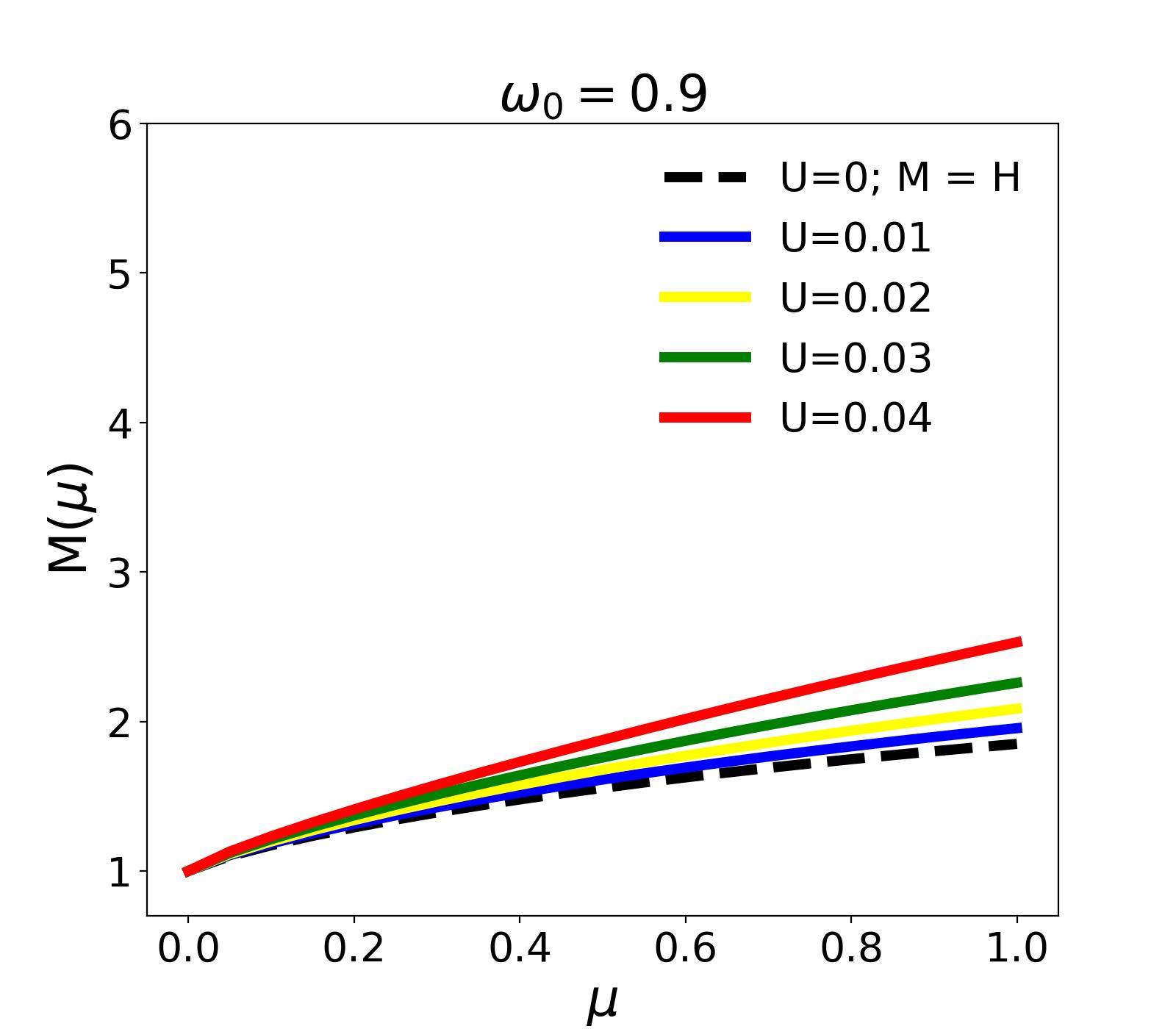}}
    \subfigure[$\tilde{\omega}_0=0.95;U_{max}=0.02$]{\includegraphics[width=0.3\textwidth]{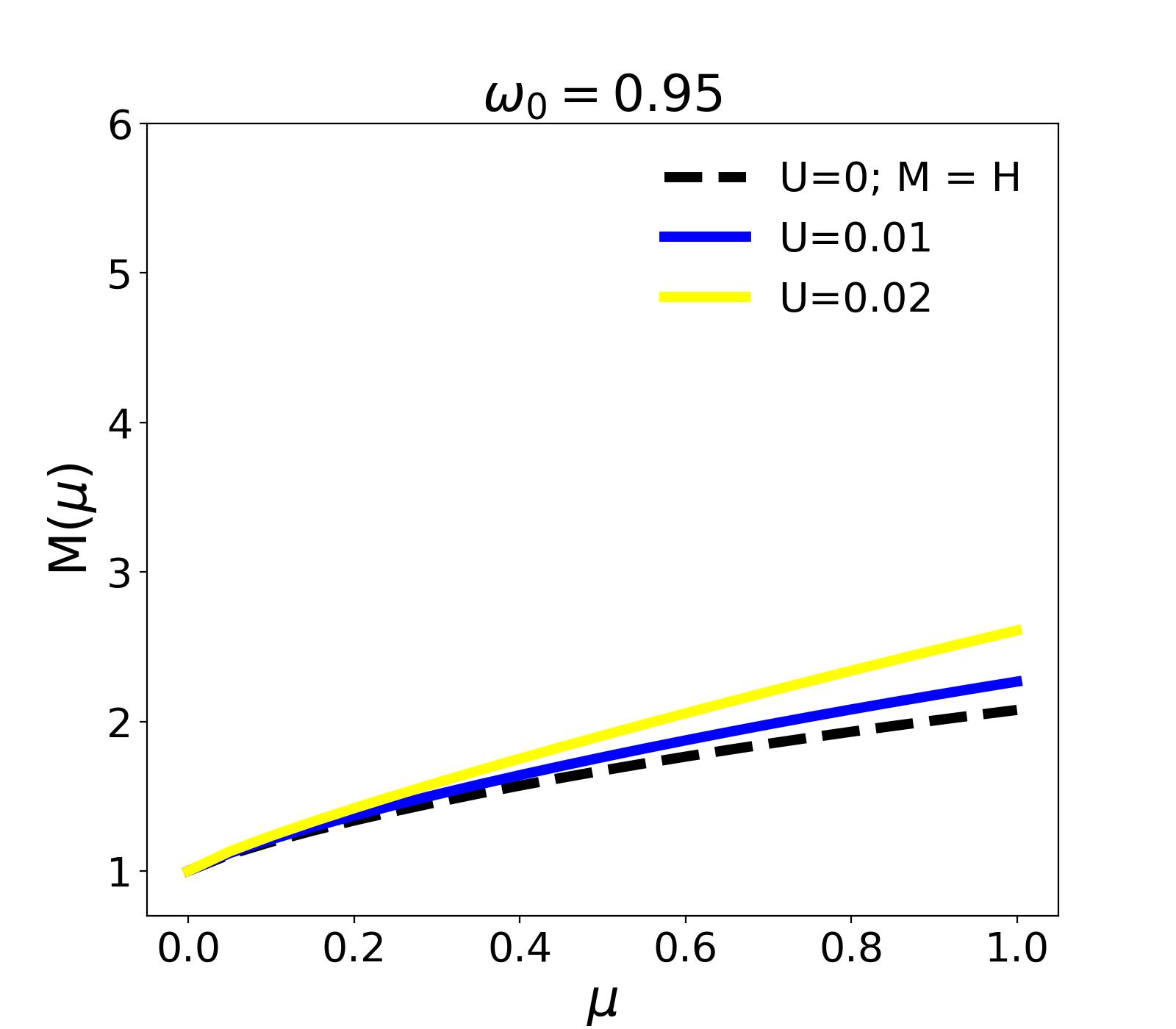}}

    \caption{Here we have plotted the function $M(\mu)$ with respect to $\mu$ for a set of single scattering albedo $\tilde{\omega}_0$ ranging from 0.0-0.95. Only for the plot~\ref{fig: M for omg=0} we have used the analytic expression given in eqn.\eqref{eq: M-function for only emission} and for the remaining plots numerical values are used. 
    In each plot, we vary the thermal emission co-efficient U within the limit of convergence as defined by the corollary \eqref{cor: 1.1}. The dashed line in each plot represents the U=0, the no thermal emission case where $M(\mu)=H(\mu)$ as discussed in sec.~\ref{sec: consistency}. It is clear from these figures that with the increase of U and $\mu$ the difference between $M(\mu)$ and $H(\mu)$ increases (see text).}
    \label{fig: Comparison plots for M function}
\end{figure}

\begin{figure}
    \centering
    {\includegraphics[width=0.3\linewidth]{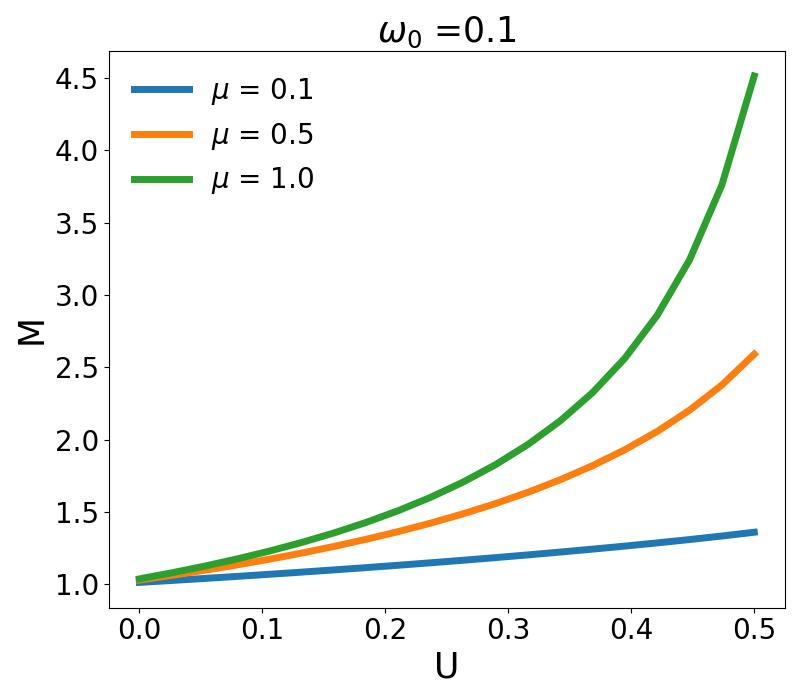}
    \includegraphics[width=0.3\linewidth]{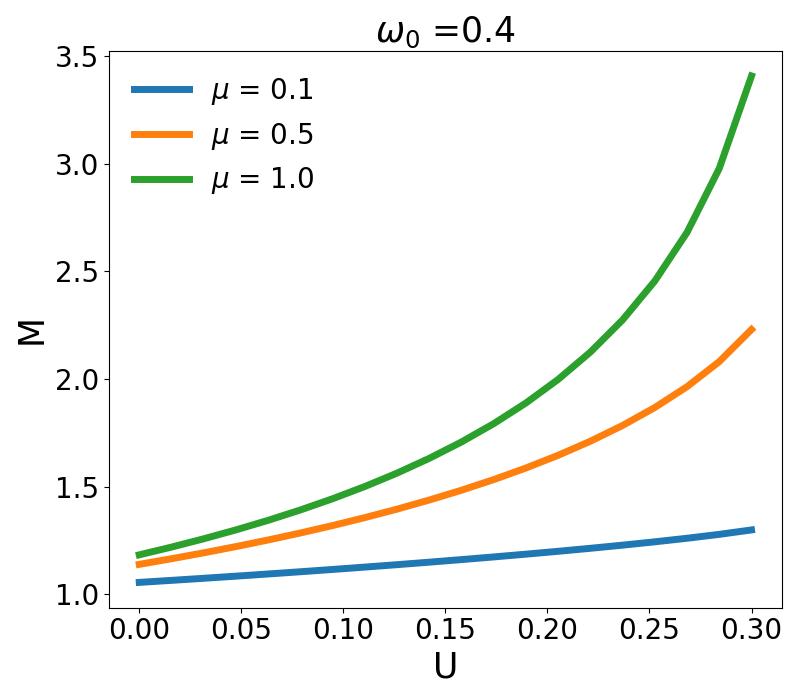}
    \includegraphics[width=0.3\linewidth]{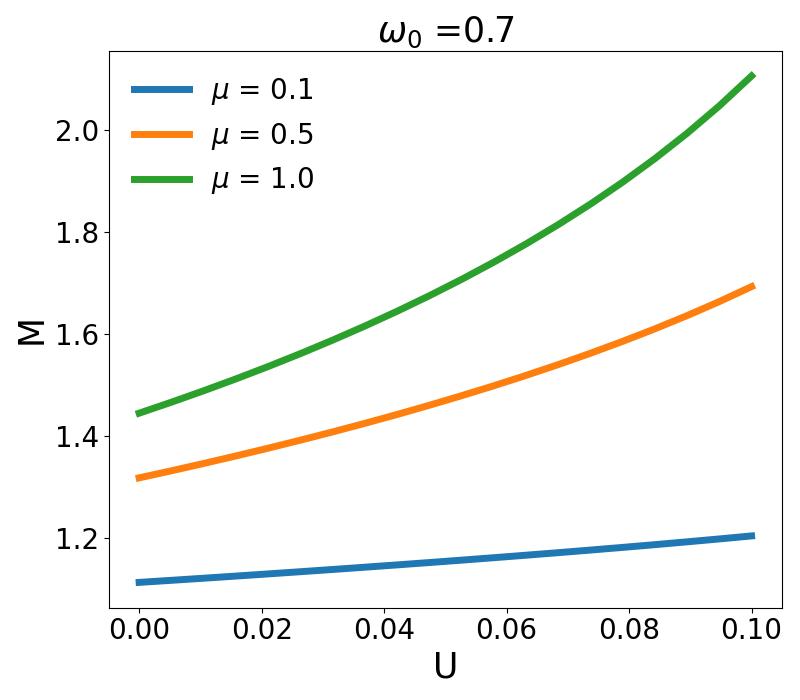}
    }
    \caption{\color{SS}{Here we have shown the variation of M with respect to thermal emission co-efficient U. These figures are  for three different single scattering albedo values  $\tilde{\omega}_0$=0.1 (low scattering), 0.4(moderate scattering), 0.7(high scattering). In each figure we plotted three different direction cosines 0.0 (tangential to the plane), 0.5 and 1.0 (perpendicular to the plane). It is clear that with increasing U the M-values also increases for $\mu\neq 0$. Also this increase became faster with larger $\mu$.}}
    \label{fig: MvsU for fixed omg}
\end{figure}

\begin{figure}
    \centering
    {\includegraphics[width=0.3\linewidth]{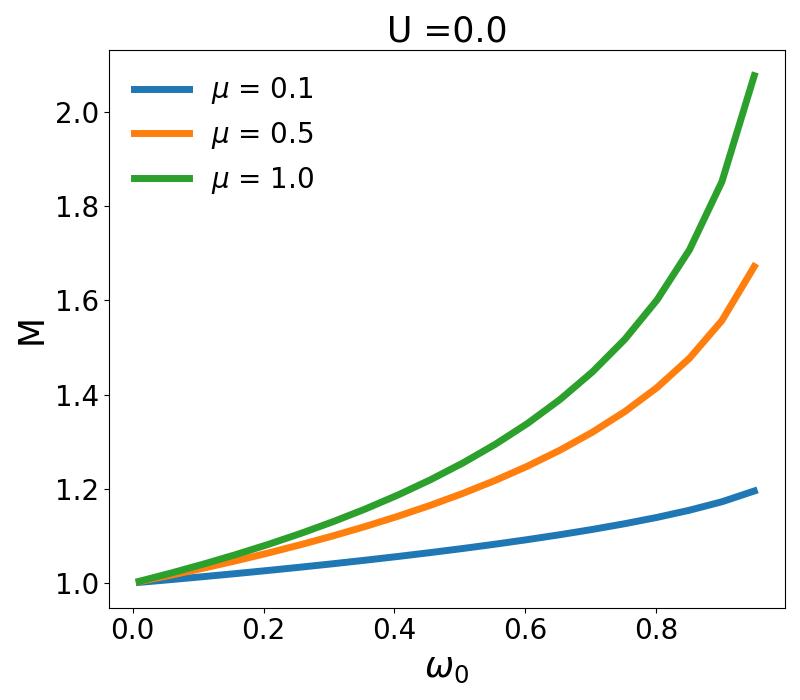}
    \includegraphics[width=0.3\linewidth]{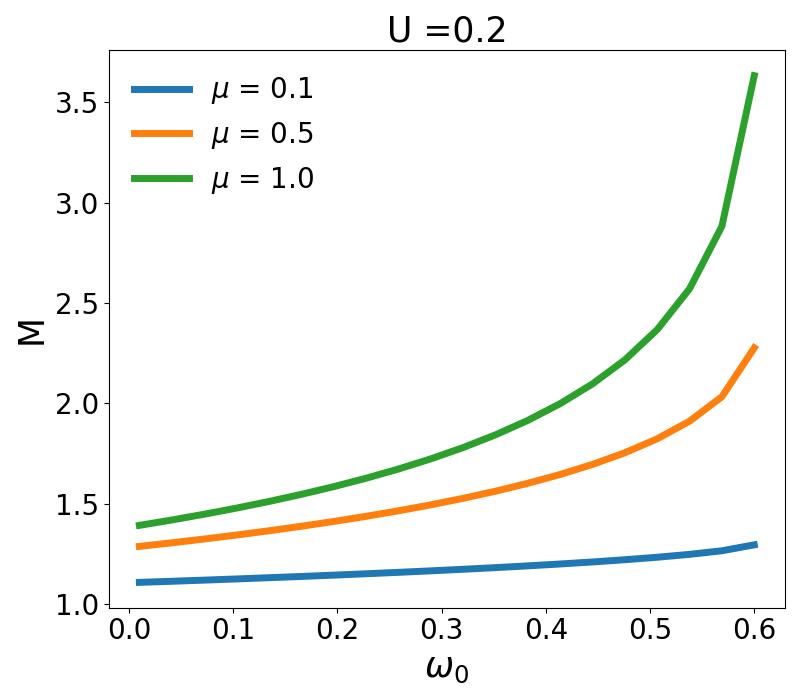}
    \includegraphics[width=0.3\linewidth]{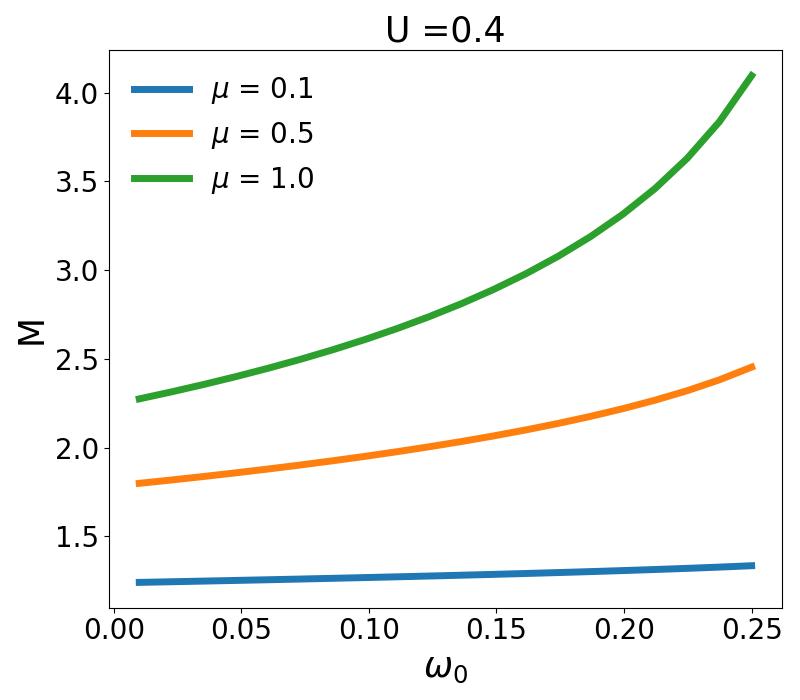}
    }
    \caption{\textcolor{MS}{Here the variation of M is studied  with respect to $\tilde{\omega}_0$. These figures are for three different U values 0, 0.2, 0.4 respectively and the variation of $\mu$ is same as in fig. \ref{fig: MvsU for fixed omg}.}\color{SS} It should be noted that left most figure actually represents the H function due to U=0 (see text)}
    \label{fig: MvsOmega for fixed U}
\end{figure}

\section{Consistency of our results}\label{sec: consistency}
We have derived the theorems and values of the $M(\mu)$ function which are more general than Chandrasekhar's $H(\mu)$-function in case of isotropic scattering, where we have added thermal emission along with the scattering. Now we show that at the low thermal emission limit i.e. $U(T)\to 0$ all of our results will match with that of the $H(\mu)$ function.

In case of {isotropic} diffusion scattering, the well known Chandrasekhar's H-function can be written \cite{Chandrashekhar} as,
\begin{equation}\label{eq: H-function}
    H(\mu) = 1+ \frac{\tilde{\omega}_0}{2}\mu H(\mu)\int_0^1 \frac{H(\mu')}{\mu + \mu'} d\mu'
\end{equation}

Now the integral theorems satisfied by $H(\mu)$ function {can be found in }\citep{Chandrashekhar,chandrasekhar1947radiative}. Also the moments of H-function are defined in \cite{Chandrashekhar} as 
$\alpha_n = \int_0^1 \mu^n H(\mu) d\mu$ {,which takes the functional form of $\alpha_0 =\frac{2}{\tilde{\omega}_0}[ 1 - \sqrt{1-\tilde{\omega}_0}]$ in case of zeroth order moment \cite{anli2021alternative}.}

In section~\ref{sec: Theorems of M-function}, we derived the integral theorems of M-function which are analogous to Chandrasekhar's H-function in case of simultaneous thermal emission and diffusely scattering atmosphere. To be consistent with Chandrasekhar's results, the theorems for M-function should match with the integral theorems of H-functions in isotropic scattering case for no thermal emission limit. 

At the limit of low thermal emission, i.e. $U\to 0; R\to 1$ {see} eqn.\eqref{eq: R-formula}. In such case the integral theorems of M-functions will reduce into,
\begin{equation*}
        \frac{\tilde{\omega_0}}{2}\int_0^1 M(\mu) d\mu = 1-[1-\tilde{\omega}_0]^\frac{1}{2}
\end{equation*}
and
\begin{equation*}
    \frac{\tilde{\omega}_0}{6}= \frac{1}{2}[\frac{\tilde{\omega}_0}{2}\int_0^1 M(\mu) \mu d\mu]^2 + [1-\tilde{\omega}_0]^\frac{1}{2}(\frac{\tilde{\omega}_0}{2})[\int_0^1 M(\mu)\mu^2 d\mu]
\end{equation*}

These matches exactly with the theorems of H-function for isotropic {scattering} case \cite{Chandrashekhar}. Also at low thermal emission the moment of $M(\mu)$ function {as given in \ref{thrm: THEOREM 1}} shows that $A_0 \to \alpha_0$ putting the limit $R\to 1$

Finally we discuss the consistency of the values of the $M(\mu)$-function as shown in fig. \ref{fig: Comparison plots for M function}.
The methods of derivation of the values has been discussed in section \ref{sec: value estimation}. In each plot the line representing U = 0.0 provides the values of $M(\mu,U=0,\tilde{\omega}_0)$ which are exactly the same as the values of $H(\mu)$ provided by \cite{chandrasekharBreen1947radiative} for the isotropic scattering case with corresponding values of $\tilde{\omega}_0$. {\color{SS} Also the left panel of figure. \ref{fig: MvsOmega for fixed U} exactly reproduces the $H(\mu,\tilde{\omega}_0)$ function in the no thermal emission limit.} Hence all of our results are fully consistent with the $H(\mu)$ function for the no thermal emission (i.e. $U\to 0$) limit. 
\section{Discussion}\label{sec: discussion}

In this work, we have derived and established a set of theorems governing the behavior of the $M(\mu, U, \tilde{\omega}_0)$ function, which incorporates both isotropic scattering and thermal emission in the context of semi-infinite diffuse reflection problem \citep{sengupta2021effects}. This marks a significant step forward by integrating the thermal emission contribution, represented by the coefficient $U(T) = \frac{B(T)}{F_{irr}}$, where $B(T)$ is the thermal emission and $F_{irr}$ is the irradiated flux, into the classical scattering framework introduced by \cite{Chandrashekhar}. The inclusion of this term modifies the traditional radiative transfer equations and introduces a new layer of complexity, which is particularly relevant in astrophysical scenarios such as exoplanets and Hot Jupiters, where internal heating significantly impacts the emergent radiation.

\subsection{Theorems and their Physical Significance}\label{subsec: Discuss theorems}

{The theorems governing the $M(\mu,U,\tilde{\omega}_0)$ function are consistent with earlier radiative transfer studies \cite{sengupta2021effects, sengupta2022atmospheric} and provide a natural generalization of Chandrasekhar’s classical diffuse scattering theory \cite{Chandrashekhar}. The key novelty of the present work lies in the explicit incorporation of thermal emission through the invariant parameter $U(T)$.}

{It is important to note that the single-scattering albedo $\tilde{\omega}_0$ and the thermal emission parameter $U(T)$ describe fundamentally different physical processes. The parameter $\tilde{\omega}_0$ governs the probability that an interaction results in scattering rather than absorption of incident radiation, whereas $U(T)$ characterizes the generation of additional photons through thermal emission. These processes do not compete within a single interaction; instead, thermal emission adds to the radiation field, which is subsequently redistributed by scattering.}

{As shown in Theorem~\ref{thrm: THEOREM 1}, the integral constraint on $M(\mu)$ reduces to the classical isotropic scattering case when thermal emission is absent ($R \to 1$), where the $\mu$-integration is fully determined by the single-scattering albedo $\tilde{\omega}_0$. This establishes the correct classical limit and validates the consistency of the formulation with earlier results for semi-infinite atmospheres (see Section~\ref{sec: consistency}).}

{When thermal emission is present ($R>1$), the integral equation for $M(\mu)$ becomes nonlinear, and the $\mu$-integration depends jointly on $\tilde{\omega}_0$ and $U(T)$. This reflects the coupled influence of scattering and internal emission on the angular redistribution of radiation, highlighting the nontrivial role of thermal emission in the diffuse reflection problem.}

\subsection{Corollaries and Their Implications}\label{subsec: Discuss Corol}

{The corollaries further constrain the behavior of the $M(\mu)$ function under different physical conditions. In particular, Corollary~\ref{cor: 1.1} defines the admissible range of $\tilde{\omega}_0$ and $U(T)$ for which physically meaningful solutions exist. For perfect scattering ($\tilde{\omega}_0=1$), thermal emission is excluded, and the problem reduces to Chandrasekhar’s classical diffuse scattering case, with $M(\mu)$ identical to $H(\mu)$. This provides an important validation of the generalized formulation.}

{Corollary~\ref{cor: 1.2} guarantees that $M(\mu)$ remains nonzero for all $\mu\in[0,1]$. In the limiting case of no scattering and no emission ($U=0$, $\tilde{\omega}_0=0$), one obtains $M(\mu,0,0)=1$, corresponding to pure transmission of the incident stellar radiation, consistent with earlier results (see Fig.~2 of \cite{sengupta2022atmospheric}). This ensures that the solution remains physically well behaved across the full parameter space.}

{Finally, Theorem~\ref{thrm: THEOREM 1} shows that the zeroth moment of $M(\mu)$ remains finite for all $\tilde{\omega}_0\le1$, ensuring a finite total reflected intensity and providing a robust basis for computing observable quantities.}

\subsection{Values}\label{subsec: Discuss Values}

{We now discuss the values of $M(\mu,U,\tilde{\omega}_0)$ obtained from Eq.~\eqref{eq: M-function} using the governing theorems. Three cases are considered: (i) only thermal emission, (ii) only scattering, and (iii) simultaneous scattering and thermal emission.}

\subsubsection{Case 1: Only Thermal Emission}\label{subsubsec: Discuss emission}

{For pure thermal emission ($\tilde{\omega}_0=0$, $U(T)\neq0$), $M(\mu)$ admits an analytic form (Eq.~\eqref{eq: M-function for only emission}). Real solutions require an upper bound on $U(T)$, given by Eq.~\eqref{eq: emission upper limit value}. Although this bound formally allows $U\lesssim0.721$, $M(\mu)$ grows rapidly near this limit; therefore, we restrict our analysis to $U\le0.6$.}

{Figure~\ref{fig: M for omg=0} shows that $M(\mu,U,0)$ increases monotonically with $\mu$, indicating stronger emission along the surface normal than at grazing angles. For fixed $\mu$, $M(\mu)$ also increases with $U(T)$, consistent with enhanced Planck emission at higher temperatures. As $U(T)$ increases, $M(\mu)$ departs increasingly from the scattering-only solution $H(\mu)$, in agreement with \citetalias{sengupta2021effects}.}

\subsubsection{Case 2: Only Scattering}\label{subsubsec: Discuss scattering}

{For pure scattering ($U=0$, $\tilde{\omega}_0\neq0$), $M(\mu)$ reduces exactly to Chandrasekhar’s $H(\mu)$ function, as discussed in Section~\ref{sec: consistency}. Values of $H(\mu,\tilde{\omega}_0)$ from \cite{chandrasekharBreen1947radiative} are shown in Figure~\ref{fig: Comparison plots for M function} (black dashed curves). As expected, increasing $\tilde{\omega}_0$ enhances the reflected intensity for all $\mu$. Corollary~\ref{cor: 1.1} further confirms that for $\tilde{\omega}_0=1$, embeded thermal emission is excluded and the results converges to \cite{chandrasekharBreen1947radiative}.}

\subsubsection{Case 3: Simultaneous Scattering and Thermal Emission}\label{subsubsec: Discuss Emission+Scattering}

{In the general case ($U\neq0$, $\tilde{\omega}_0\neq0$), $M(\mu)$ exhibits a nonlinear dependence on both parameters. For fixed $\tilde{\omega}_0$, $M(\mu)$ increases with $U$ (Fig.~\ref{fig: MvsU for fixed omg}), while for fixed $U$, it increases with $\tilde{\omega}_0$ (Fig.~\ref{fig: MvsOmega for fixed U}). This reflects the enhanced scattering of thermally emitted Planck photons at higher albedo, leading to an amplified reflected intensity.}

{At grazing angles ($\mu\to0$), $M(\mu)$ converges to unity regardless of $U$ or $\tilde{\omega}_0$, indicating that emission and scattering effects become indistinguishable in this limit. The strongest deviations from $H(\mu)$ occur near normal incidence ($\mu=1$), where both thermal emission and scattering are most effective.}

{Overall, $M(\mu,U,\tilde{\omega}_0)$ increases monotonically with $\mu$, $U$, and $\tilde{\omega}_0$ (Fig.~\ref{fig: Comparison plots for M function}), demonstrating the coupled role of scattering and thermal emission in shaping the diffusely reflected specific intensity.}

{We therefore conclude that $M(\mu,U,\tilde{\omega}_0)$ represents a natural generalization of Chandrasekhar’s $H(\mu)$ function. It retains the classical behavior in the appropriate limits (Section~\ref{sec: consistency}) while extending the theory to atmospheres with invariantly embedded thermal emission.}

\subsection{Limitations and Future work}\label{subsec: limitation}
This work is the first attempt to quantitatively establish the thermal emission contribution for the semi-infinite diffuse reflection problem for the isotropic scattering case.
While this work presents a comprehensive analysis of the $M(\mu, U, \tilde{\omega}_0)$ function, there are a few limitations that should be highlighted, which can be addressed in future research.

Firstly, we have treated the thermal emission coefficient $U$ as a variable parameter rather than explicitly incorporating temperature as a direct variable. {Although this approach is effective for many practical applications, particularly when temperature gradients are weak, a more detailed treatment with an explicit temperature profile would be required for systems with complex thermal structures where temperature-dependent effects strongly influence radiative transfer.}

Secondly, in the conservative case of perfect scattering, where $\tilde{\omega}_0 = 1$, numerical solutions for the $M(\mu, U, \tilde{\omega}_0)$ function were achieved only for U=0 case within the parameter range considered in this study. In this limit, thermal emission becomes impossible, and the system reduces to the purely scattering case, which is well understood in the literature. However, further work could explore methods for addressing this specific case and ensure smooth transitions between scattering and thermal emission regimes, enhancing the flexibility of the model for various astrophysical environments.

Despite these limitations, this work provides a solid foundation for the calculation of the $M(\mu, U, \tilde{\omega}_0)$ function, with numerous opportunities for refinement. These challenges offer pathways for future investigations, including incorporating temperature dependencies and handling the extreme limits of perfect scattering, which would further extend the applicability of our model in diverse astrophysical and planetary contexts.

One promising direction is the extension of this framework to more complex atmospheric scenarios. Specifically, the $M(\mu,U,\tilde{\omega}_0)$ function, derived for the semi-infinite atmosphere, has a direct counterpart in the $V(\mu)$ function, which is applicable to finite atmospheres as shown in \cite{sengupta2022atmospheric}. Following the approach outlined in this study, the $V(\mu)$ function can also be derived and calculated for various cases of isotropic scattering and thermal emission. This would enable a more comprehensive treatment of radiative transfer in both finite and semi-infinite atmospheric contexts, providing greater applicability to a range of astrophysical environments.

In the realm of exoplanetary science, atmospheric spectra for various exoplanets, particularly Hot Jupiter-type planets, have been modeled to match observational data, including reflection, transmission, and emission spectra \cite{tinetti2013spectroscopy}. For such planets, incorporating the thermal emission contribution in the form of the $M(\mu)$ function is crucial. The exoplanetary atmosphere is modeled directly using the $H(\mu)$ function considering only the scattering phenomena \cite{madhusudhan2012analytic}. Extending the present work to model the atmospheric emission spectra of exoplanets during the secondary eclipse \citep{hansen2008absorption}, accounting for both scattering and thermal emission, would enhance our ability to interpret observed data and improve our understanding of the atmospheric composition and dynamics of these distant worlds.

Additionally, the study of tidally locked gas giants offers another exciting area for future exploration. These planets experience significant day-night temperature contrasts, which influence atmospheric circulation patterns and heat redistribution. The {day-side} emission spectra of such planets, when combined with the scattering effects considered here, could provide deeper insights into their atmospheric structure and behavior. Studies of this nature would benefit from the inclusion of the $M(\mu,U,\tilde{\omega}_0)$ function in models of planetary emission spectra, as discussed in \cite{sengupta2023atmospheric}. The interplay between emission and scattering in tidally locked exoplanets represents a complex yet fascinating avenue for further research, with important implications for the habitability of these distant worlds.
\subsection{Applicability in case of exoplanetary system}
{\color{SS}In this work, we find that the contribution of thermal emission becomes significant depending on the dimensionless parameter \( U \), defined as \( U = \frac{B(T)}{F} \), where \( B(T) \) is the Planck function and \( F \) is the incident flux. Here we have studied an example case for exoplanetary atmosphere to apply {this} theory. For an exoplanetary atmosphere, the incident stellar flux at a given wavelength \( \lambda \) can be approximated by
\begin{equation}
F_{\lambda}^{\text{irr}} = \left(\frac{R_*}{D}\right)^2 B_{\lambda}^{\text{star}}(T_*)
\end{equation}

where \( T_*, R_* \), and \( D \) denote the stellar temperature, stellar radius, and the star–planet separation, respectively \cite{guillot2010radiative}. Assuming both the star and the planet behave as blackbodies, the planetary thermal emission at the same wavelength can be written as \( B_{\lambda}^p(T_{\text{eq}}) \), where \( T_{\text{eq}} \) is the planet's equilibrium temperature. The thermal emission coefficient at that wavelength is then given by

\begin{equation}\label{eq:U_exoplanet}
U_\lambda = \left(\frac{D}{R_*}\right)^2 \frac{B_{\lambda}^p(T_{\text{eq}})}{B_{\lambda}^{\text{star}}(T_*)}
\end{equation}
In this study, we have shown that the coupled emission and scattering regime is physically meaningful only for values of \( U < 0.7{\color{MS}21} \), within which the function \( M(\mu, U_\lambda, \tilde{\omega}_0) \) converges. Moreover, the larger the value of \( U \) within this convergence range, the more pronounced the thermal emission effect along with scattering.

As a case study, we consider the exoplanet K2-137b, an Earth-sized planet orbiting an M-dwarf star. The parameters are as follows: \( T_{\text{eq}} = 1471\,\text{K} \), \( T_* = 3492\,\text{K} \), \( R_* = 0.442\,R_\odot \), and orbital distance \( D = 0.0058\,\text{AU} \), as reported in \cite{smith2018k2}. Incorporating this system parameters we compute and plot  Fig.~\ref{fig:K2137b}  which shows the variation of the thermal emission coefficient \( U_\lambda \) as a function of wavelength for the ultra-short-period exoplanet K2-137b, assuming blackbody spectra for both the planet and its host M-dwarf star. The resulting curve demonstrates a monotonic increase in \( U_\lambda \) with wavelength.

Three distinct radiative regimes are indicated in the plot:

\begin{itemize}
    \item \textit{Scattering-dominated regime} (\( \lambda < 0.85\,\mu\mathrm{m} \), shaded blue): Here, \( U_\lambda \ll 1 \), implying that the incoming stellar irradiation dominates the radiative budget. The thermal emission from the planet is negligible, and the radiative transfer is governed primarily by scattering processes. In this region the diffuse reflection theory of \cite{Chandrashekhar} can be used without any modification.

    \item \textit{Simultaneous emission–scattering regime} (\( 0.85\,\mu\mathrm{m} < \lambda < 2.5\,\mu\mathrm{m} \), shaded pink): This region corresponds to \( 0.1 \lesssim U_\lambda \lesssim 0.7 \), where both thermal emission and scattering contribute significantly. Within this range, the function \( M(\mu, U_\lambda, \tilde{\omega}_0) \) has been shown to converge, marking it as the valid regime for the theoretical model of coupled emission and scattering proposed in \cite{sengupta2021effects,sengupta2022atmospheric} and developed in this work.

    \item \textit{Emission-dominated regime} (\( \lambda > 2.5\,\mu\mathrm{m} \), shaded orange): In this domain, \( U_\lambda > 0.7 \), suggesting that thermal emission dominates the radiative transfer. The \( M(\mu, U_\lambda, \tilde{\omega}_0) \) diverging in this region, and the radiative transfer is dominated by emission. Hence diffuse reflection theory loses it's meaning and applicabilty in the region and further.
\end{itemize}

The spectral window between \( 0.85\,\mu\mathrm{m} \) and \( 2.5\,\mu\mathrm{m} \) is thus identified as the optimal regime for validating the theoretical framework under conditions where both scattering and emission are relevant. In the observational prospect, several instruments like JWST-NIRSpec, HST-WFC3, ARIEL enable potential empirical validation of the emission–scattering model across a range of hot exoplanets, with K2-137b serving as a promising initial candidate. A comprehensive statistical study, observational sensitivity for multiple exoplanets, however, is out of the scope of the present article and left for future work.

\begin{figure}[h]
\centering
\includegraphics[width=0.7\textwidth]{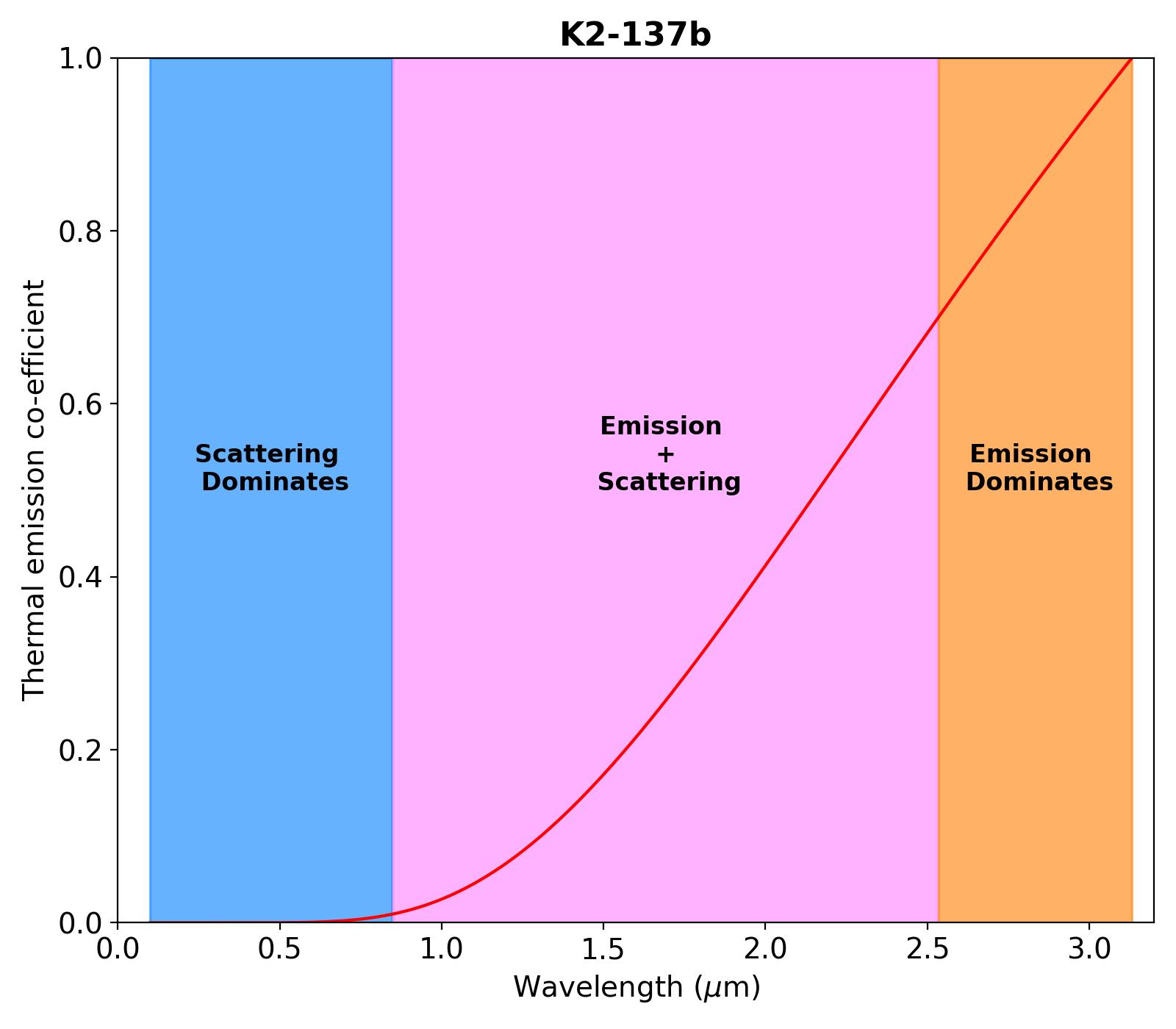}
\caption{Thermal emission coefficient U as a function of wavelength for the exoplanet K2-137b using the formulation eqn.\eqref{eq:U_exoplanet}. The plot illustrates three radiative regimes: a scattering-dominated regime ($U\to 0$
 ), a transitional region where both thermal emission and scattering contribute ($0\leq U<0.7$), and an emission-dominated regime ($U>0.7$). These divisions are essential for identifying spectral regions where the theoretical framework of coupled emission and scattering remains valid.}
\label{fig:K2137b}
\end{figure}

}

\section{Conclusion}\label{sec: conclusion}

{This work provides a quantitative extension of Chandrasekhar’s diffuse reflection theory to atmospheres where thermal emission and scattering act simultaneously. The formulation retains the correct classical limit while introducing a physically consistent framework for thermally active systems.}
The theorems and values presented here offer a robust mathematical framework for modeling radiative transfer in atmospheres where both scattering and thermal emission play significant roles. This is particularly important for astrophysical phenomena involving exoplanets, Hot Jupiters, and brown dwarfs, where both thermal emission and scattering contribute comparably to the final emergent radiation. Our results directly address the need for accurate modeling of the reflected specific intensity in such systems, which is crucial for interpreting transmission and reflection spectra from observational data.

Moreover, the ability to quantify the $M(\mu,U,\tilde{\omega}_0)$ function and its moments opens the door for more precise predictions of radiative properties in thermally active atmospheres. This is expected to have important implications for the study of habitability, energy balance, and atmospheric composition in exoplanet research, as well as for more general applications in planetary science and astrophysics.
Overall, these potential research directions open new opportunities to refine and expand upon the current understanding of radiative transfer in astrophysical and planetary atmospheres. By addressing these complex cases and incorporating additional physical processes, future studies can provide a more holistic view of the radiative properties of exoplanets and other celestial bodies.
In conclusion, the theorems developed in this work provide a comprehensive, generalized framework for understanding radiative transfer in semi-infinite atmospheres with both scattering and thermal emission, making a significant contribution to the field of radiative transfer theory and its applications in astrophysics and beyond.

\textit{Acknowledgement:}
{Soumya Sengupta (SS) acknwoledges the support by the European Research Council (ERC) under the European Union’s Horizon 2020 research and innovation program (ERC Starting Grant ``IMAGINE'' No. 948582, PI: Daniele Vigan\`{o}) and  by the Spanish program Unidad de Excelencia María de Maeztu, awarded to the Institute of Space Sciences (ICE-CSIC), CEX2020-001058-M. SS is grateful to Daniele Vigan\`{o} for the insightful suggestions to improve this theoretical work up to the level of observation. Also thanks to Ignasi Ribas and Juan Carlos Morales for their insightful comments. 
MS: This work was supported by the Institute for Basic Science (IBS-R035-C1).
We additionally thank the referee for a careful reading of the manuscript and for constructive comments that significantly improved the clarity and presentation of this work.}

\newpage
\appendix

\section{{Alternate proof of the zeroth moment}}\label{appendix}

{Here we present an alternate way to prof \textit{Theorem} \ref{thrm: THEOREM 1}. This method is more direct than that presented in the main text.}

\begin{proof}
{For n=0, eqn.\eqref{eq: M-moments} can be represented as,} 
    \begin{equation}
    \begin{split}
        A_0 = \int_0^1 M(\mu) d\mu=
        1 + 2U(T)\int_0^1 M(\mu)\mu \log(1+\frac{1}{\mu}) d\mu + \frac{\tilde{\omega_0}}{2}\int_0^1 \int_0^1 \frac{\mu}{\mu+\mu'}M(\mu)M(\mu')d\mu'
    \end{split}
\end{equation}

{Switching $\mu$ and $\mu'$ in the last expression and taking average we will get,}

\begin{equation}
    \begin{split}
        A_0 =&
        1 + 2U(T)\int_0^1 M(\mu)\mu \log(1+\frac{1}{\mu}) d\mu + \frac{\tilde{\omega_0}}{4}\int_0^1 \int_0^1 M(\mu)M(\mu')d\mu'd\mu\\
        =&
        1  + 2U(T)\int_0^1 M(\mu)\mu \log(1+\frac{1}{\mu}) d\mu + \frac{\tilde{\omega_0}}{4}A_0^2\\
        =&
        R + \frac{\tilde{\omega_0}}{4}A_0^2
    \end{split}
\end{equation}

{Now this quadratic equation has the following possible solutions,}
\begin{equation*}
    \frac{\Tilde{\omega}_0}{2} A_0 = 1 \pm \sqrt{1 - \Tilde{\omega}_0R}
\end{equation*}

 {It is evident that the left hand side will uniformly converges to zero with $\Tilde{\omega}_0$. This is possible only if we choose the negative sign in the right hand side. Hence the equation for $A_0$ will become,}
 \begin{equation}
     A_0 = \frac{2}{\Tilde{\omega}_0}(1 - \sqrt{1 - \Tilde{\omega}_0R})
 \end{equation}
{ This is the same as eqn.\eqref{eq: theorem 1}. Hence proved \textit{Theorem} \ref{thrm: THEOREM 1}.}
\end{proof}
\newpage
{\software{NumPy \citep{2020NumPy-Array}, SciPy \citep{2020SciPy-NMeth}}}


\bibliography{M_function_paper}
\bibliographystyle{aasjournal}
\end{document}